\documentclass[shortAfour,times,sageh]{sagej}

\usepackage{moreverb,url}
\usepackage{graphics} 
\usepackage{epstopdf}

\usepackage{times} 
\usepackage{amsmath} 
\usepackage{amssymb}  
\usepackage{amsthm}
\usepackage{caption}
\usepackage{subcaption}
\usepackage{multirow}
\usepackage{color}
\usepackage{algorithm}
\usepackage[noend]{algpseudocode}
\usepackage{microtype}

\usepackage{graphicx}
\usepackage[inline]{enumitem}

\graphicspath{{./figures/}}

\newcommand\BibTeX{{\rmfamily B\kern-.05em \textsc{i\kern-.025em b}\kern-.08em
T\kern-.1667em\lower.7ex\hbox{E}\kern-.125emX}}

\def\volumeyear{2016}

\algrenewcommand\algorithmicindent{0.4cm}
\DeclareMathOperator*{\argmin}{arg\,min}

\newtheorem{theorem}{Theorem}
\newtheorem{lemma}{Lemma}

\theoremstyle{definition}
\newtheorem{definition}{Definition}
\newtheorem{assumption}{Assumption}
\newtheorem{problem}{Problem}

\newcommand{\lowx}{\mathbf{x}^a}
\newcommand{\upperx}{\mathbf{x}^b}
\newcommand{\allset}{\mathcal{C}\cup\bar{\mathcal{C}}}
\newcommand{\uncontrolled}{\bar{\mathcal{C}}}

\newcommand{\aggsignal}[1]{\mathbf{#1}(\cdot)}
\newcommand{\signal}[2]{{#1}_{#2}(\cdot)}
\newcommand{\aggdesire}[1]{\mathbf{u}^{{#1},\infty}_{desire}(\cdot)}

\newcommand{\aggsafe}[1]{\mathbf{u}^{{#1}}_{safe}(\cdot)}
\newcommand{\safe}[2]{{u}^{{#1},\infty}_{safe,{#2}}(\cdot)}
\newcommand{\aggout}[1]{\mathbf{u}^{{#1}}_{out}(\cdot)}

\newcommand{\procedureS}{\textsc{Scheduling}}
\newcommand{\procedureES}{\textsc{Exact}}
\newcommand{\procedureAS}{\textsc{Approx}}
\newcommand{\procedureRES}{\textls[-85]{\textsc{RelaxedExact}}}
\newcommand{\supervisor}{\textsc{Supervisor}}
\newcommand{\efficientS}{\textsc{EffSupervisor}}
\newcommand{\garey}{\textsc{Polynomial}}
\begin{document}

\runninghead{Supervisors for Intersection Collision Avoidance}

\title{Robust Supervisors for Intersection Collision Avoidance in the Presence of Uncontrolled Vehicles}

\author{Heejin Ahn\affilnum{1}, Andrea Rizzi\affilnum{2}, Alessandro Colombo\affilnum{3}, and Domitilla Del Vecchio\affilnum{1}}

\affiliation{\affilnum{1}Department of Mechanical Engineering, Massachusetts Institute of Technology, Cambridge, MA, USA\\
\affilnum{2}Tri-Institutional Program in Computational Biology and Medicine, Weill Cornell Medical College, New York, NY, USA\\
\affilnum{3} Dipartimento di Elettronica Informazione e Bioingegneria, Politecnico di Milano, Milano, Italy}
\corrauth{Heejin Ahn, Department of Mechanical Engineering, Massachusetts Institute of Technology, 77 Massachusetts Avenue, Cambridge, MA 02139, USA.}

\email{hjahn@mit.edu}

\begin{abstract}
We present the design and validation of a centralized controller, called a \textit{supervisor}, for collision avoidance of multiple human-driven vehicles at a road intersection, considering measurement errors, unmodeled dynamics, and uncontrolled vehicles. We design the supervisor to be least restrictive, that is, to minimize its interferences with human drivers. This performance metric is given a precise mathematical form by splitting the design process into two subproblems: verification problem and supervisor-design problem. The verification problem determines whether an input signal exists that makes controlled vehicles avoid collisions at all future times. The supervisor is designed such that if the verification problem returns yes, it allows the drivers' desired inputs; otherwise, it
overrides controlled vehicles to prevent collisions. As a result, we propose \textit{exact} and \textit{efficient} supervisors. The exact supervisor solves the verification problem exactly but with combinatorial complexity. In contrast, the efficient supervisor solves the verification problem within a quantified approximation bound in polynomially bounded time with the number of controlled vehicles. We validate the performances of both supervisors through simulation and experimental testing. 
\end{abstract}

\maketitle
\keywords{Class file, \LaTeXe, \textit{SAGE Publications}}

%

\section{Introduction}

Autonomous robots have drawn attention in various applications, such as exploring unknown environment (\cite{maimone_two_2007}), moving materials in warehouses (\cite{wurman_coordinating_2008}), and collecting data on ocean conditions (\cite{smith_planning_2010}). Recently, there has been extensive research on autonomous vehicles from academic, industrial, and governmental sectors for the purpose of reducing the number of traffic accidents. Research has focused on developing fully autonomous vehicles as well as improving safety of human-driven vehicles by means of newly available automation, sensing, and communication capabilities. A major obstacle to the development of collision avoidance architecture for large traffic networks is computational complexity.

In general, there have been several approaches to reduce computational complexity in collision avoidance for multiple vehicles. In a decentralized framework, each vehicle makes its own decision to avoid collisions with its neighboring vehicles, thereby dividing a large problem into smaller local problems. To design a decentralized control law, \cite{hoffmann_decentralized_2008} and \cite{gillula_applications_2011} used reachability analysis for hybrid systems, and \cite{mastellone_formation_2008} defined a potential function. While computationally efficient, this framework is usually more conservative and unable to prevent a deadlock, where no further control input exists to terminate processes (\cite{cassandras_introduction_2008}). A centralized framework considers a whole system and thus can be less conservative, but computationally demanding. Collision avoidance problems were formulated into mixed integer linear programming (MILP) by assuming discrete-time linear vehicle dynamic models (\cite{richards_spacecraft_2002,borrelli_milp_2006}) and considering geometric construction of collisions with vehicle models of instantaneous speed or angle changes (\cite{pallottino_conflict_2002,alonso-ayuso_collision_2011}). These MILP formulations are then solved by commercially available software. 
In this paper, we present a centralized controller, in which the collision avoidance problem is translated into a scheduling problem and then solved by solving this scheduling problem. Moreover, we provide an approximate solution of this problem.

In collision avoidance problems at road intersections, the complexity can be mitigated to some extent by exploiting the fact that vehicles follow predefined paths, and side-impacts can be avoided by approximately scheduling their time of occupancy of the shared intersection (\cite{peng_convexity_2005,peng_coordinating_2005}). Based on this concept, several autonomous intersection management schemes have been studied. \cite{kowshik_provable_2011} proposed a hybrid architecture comprising an interplay between centralized coordination and distributed agents. Centralized coordination assigns time slots to agents, and distributed agents determine if they can cross an intersection within allocated time slots while avoiding rear-end collisions. \cite{wu_cooperative_2012} employed an ant colony algorithm as an approximate solution for finding an optimal sequences of vehicles to improve traffic efficiency. Collision avoidance problems were formulated into nonlinear constrained optimization to eliminate overlaps of 
given trajectories inside an intersection (\cite{lee_development_2012}) and to minimize the risk of collision (\cite{kamal_vehicle-intersection_2014}). These works consider fully autonomous vehicles and can solve collision avoidance problems by finding one safe input. However, when human operators drive vehicles, a controller needs to be least-restrictive, that is, override human drivers only when they can cause a collision.


To ensure least restrictiveness, all possible inputs must be taken into account, usually at expense of computational cost. Moreover, this exhaustive evaluation must be done frequently because controllers must keep monitoring vehicles' safety and intervene only when drivers are unable to prevent collisions.  In \cite{hafner_computational_2011,hafner_cooperative_2013}, a safety control was designed for two vehicles. Their controller was validated in laboratory experiments (\cite{hafner_computational_2011}) and field experiments (\cite{hafner_cooperative_2013}). \cite{colombo_efficient_2012} translated a collision detection problem to a scheduling problem and employed a scheduling algorithm to design a safety control for multiple vehicles. Rear-end collisions as well as intersection collisions were considered in \cite{colombo_least_2014}.

In this paper, we propose a least-restrictive controller, called a \textit{supervisor}, that prevents intersection collisions among human-driven vehicles. 
Our work extends the result of \cite{colombo_efficient_2012} in that 1) we consider sources of uncertainty, including measurement errors, unmodeled dynamics, and uncontrolled vehicles, and 2) we perform a lab-based experiment to validate the supervisor in a setting subject to many sources of uncertainty. Here, controlled vehicles communicate with and are controlled by the supervisor, whereas uncontrolled vehicles are not. The inclusion of uncontrolled vehicles accounts for a realistic mixed-traffic scenario where unequipped vehicles still travel on roads. 


To design a supervisor, we formulate two problems: verification problem and supervisor-design problem. The verification problem determines the existence of an input signal that makes controlled vehicles avoid all future collisions. We prove that this problem is equivalent to an Inserted Idle-Time (IIT) scheduling problem, where an inserted idle-time represents a set of time intervals during which uncontrolled vehicles can occupy the intersection. This formulation was introduced in \cite{ahn_supervisory_2014} to account for uncontrolled vehicles under perfect measurement and dynamic models. The supervisor is designed to override the drivers of controlled vehicles only when the verification problem determines that there will be no input signal to avoid collisions.

In order to study the trade-off between exactness and computational efficiency, we propose \textit{exact} and \textit{efficient} supervisors. The exact supervisor solves the verification problem exactly, and thus, overrides drivers only when strictly necessary. However, since the verification problem has combinatorial complexity, this supervisor is not scalable with the number of controlled vehicles. To improve the computational efficiency, the efficient supervisor is designed to solve the verification problem within a quantified approximation bound in polynomially bounded time. We validate the supervisors by performing computer simulations and lab-based experimental testing. Some of these experimental results were presented in~\cite{ahn_experimental_2015}.

This paper is organized as follows. In Section~\ref{section:systemdefinition}, we define an intersection model and a vehicle dynamic model. In Section~\ref{section:problemstatement}, we formulate two problems: verification problem and supervisor-design problem, which are solved exactly in Section~\ref{section:exactsol} and approximately in Section~\ref{section:approxsol}. The simulation results are given in Section~\ref{section:simulation}, and the experimental results in Section~\ref{section:experiment}.

\section{System Definition}\label{section:systemdefinition}
In this section, we introduce an intersection model and the vehicle dynamics.
\subsection{Intersection model}
\begin{figure}[htb!]
	\centering
	\includegraphics[width=\columnwidth,keepaspectratio=true,angle=0]{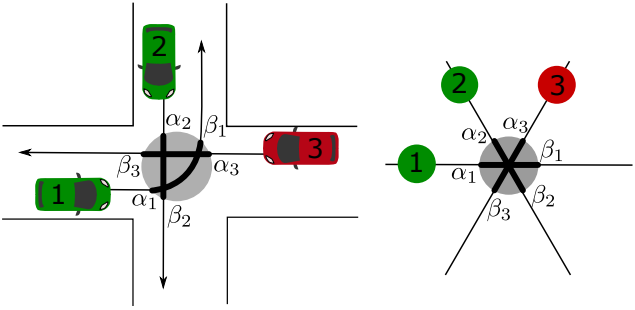}
	\caption{Road intersection (left) and its simplified model (right). We assume that the vehicles follow the prescribed paths, which intersect at one point. An intersection is modeled as an area containing this point (the shaded area), and its location is denoted by an open interval $(\alpha_i, \beta_i)$ along the path of vehicle $i$.}
	\label{figure:generalcars}
\end{figure}

Consider $n$ human-driven vehicles approaching an intersection along
$n$ different paths. As shown in Figure \ref{figure:generalcars}, the
intersection is modeled as an area containing the points at which these paths intersect. An open
interval $(\alpha_i,\beta_i)\subset \mathbb{R}$ denotes the location of this
conflict area on the longitudinal path of vehicle $i$. Among $n$
vehicles, $n_c$ vehicles are communicating with a supervisor, which takes control of these vehicles when potential crashes are detected and returns control back to the drivers when there is no more threat. The other $n_{\bar{c}}=n-n_c$ vehicles are not communicating with, and can never be controlled by the supervisor. The supervisor however can measure their positions and speeds and use these measurements to predict their possible future behaviors. We define a controlled set $\mathcal{C}$ as a set of controlled vehicles and an uncontrolled set $\bar{\mathcal{C}}$ as a set of uncontrolled vehicles. For notational simplicity, we number the vehicles such that $\mathcal{C}=\{1,2,\ldots,n_c\}$ and $\uncontrolled=\{n_c+1,n_c+2,\ldots,n\}$.

\subsection{Vehicle model}
To describe the longitudinal dynamics of vehicle $i\in \allset$, we
introduce a state $x_i(t)=(y_i(t), v_i(t))\in X_i\subset \mathbb{R}^2$, where $y_i(t)\in
Y_i\subseteq \mathbb{R}$ and $v_i(t)\in[v_{i,min},v_{i,max}]\subset\mathbb{R}$
are the position and the speed along the longitudinal path. Let $d_i(t)=(d_{y,i}(t),d_{v,i}(t))\in\mathbb{R}^2$ denote a
disturbance, where $d_{y,i}(t)\in[d_{y,i,min}, d_{y,i,max}]$ and $d_{v,i}(t)\in [d_{v,i,min}, d_{v,i,max}]$ account for unmodeled dynamics of $y_i(t)$ and $v_i(t)$, respectively. The longitudinal dynamics of controlled vehicle $j\in\mathcal{C}$ and uncontrolled vehicle $\gamma\in\uncontrolled$ are modeled as 
\begin{align}\label{equation:generalmodel}
\dot{x}_j=F_j(x_j,u_j,d_j), &&
\dot{x}_\gamma=F_\gamma(x_\gamma,w_\gamma,d_\gamma),
\end{align}
where $u_j\in U_j:=[u_{j,min},u_{j,max}]\subset \mathbb{R}$ is a throttle or brake input to controlled
vehicle $j$ applied by a supervisor, and $w_\gamma\in [w_{\gamma,min},w_{\gamma,max}]\subset\mathbb{R}$ is a driver-input to uncontrolled vehicle $\gamma$.

The parallel composition of \eqref{equation:generalmodel} is denoted as follows:
$$\dot{\mathbf{x}} = \mathbf{F(x,u,w,d)},$$ where
$\mathbf{x}=(x_1,\ldots,x_{n_c},x_{n_c+1},\ldots,x_{n})\in \mathbf{X}$, $\mathbf{u} =(u_1,u_2,\ldots,u_{n_c})\in \mathbf{U}$, $\mathbf{w} =
(w_{n_c+1},w_{n_c+2},\ldots,w_{n})$, and $\mathbf{d} =
(d_1,\ldots,d_{n_c},d_{n_c+1},\ldots,d_{n})$. The output of the system is the position of all vehicles, denoted by $\mathbf{y}=(y_1,\ldots,y_{n_c},y_{n_c+1},\ldots,y_n)\in \mathbf{Y}$. 

Let $\signal{u}{j}\in \mathcal{U}_j$ and $\signal{w}{\gamma}\in\mathcal{W}_\gamma$ denote an input signal to controlled vehicle
$j\in\mathcal{C}$ and a driver-input signal to uncontrolled vehicle
$\gamma\in\uncontrolled$, respectively. Let $\signal{d}{i}=(\signal{d}{y,i},
\signal{d}{v,i})\in\mathcal{D}_i$ denote a disturbance signal of vehicle
$i\in\allset$. We say $\signal{u}{j}\leq u'_{j}(\cdot)$ if $u_{j}(t)\leq u'_{j}(t)$ for all $t$.
Similarly, $\signal{w}{\gamma}\leq w'_{\gamma}(\cdot)$ if $w_{\gamma}(t)\leq
w'_{\gamma}(t)$ for all $t$, and $\signal{d}{i}\leq d'_{i}(\cdot)$ if
$d_{y,i}(t)\leq d'_{y,i}(t)$ and $d_{v,i}(t)\leq d'_{v,i}(t)$ for all $t$.


 Let $x_j({t},\signal{u}{j},\signal{d}{j}, x_j(0))$ denote the state of controlled vehicle $j$ at time $t$ with an
input signal $\signal{u}{j}$ and a disturbance signal $\signal{d}{j}$ starting at an initial state $x_j(0)$. Similarly, let $x_\gamma(t,\signal{w}{\gamma},\signal{d}{\gamma},x_\gamma(0))$ denote the state of uncontrolled vehicle $\gamma$ at time $t$ with a
driver-input signal $\signal{w}{\gamma}$ and a disturbance signal
$\signal{d}{\gamma}$ starting at $x_\gamma(0)$. Unless the input signals and initial state are important, we write $x_i(t)$ and $x_\gamma(t)$. We say $x_i(t)\leq x'_i(t)$ if $y_i(t)\leq y'_i(t)$ and $v_i(t)\leq v'_i(t)$ for all $t$ for $i\in\allset$ and make the following assumptions.

\begin{assumption}\label{assumption:x_orderpreserving}
	The functions $F_j$ and $F_\gamma$ in \eqref{equation:generalmodel} are order-preserving. That is, if
	$\signal{u}{j}\leq u'_j(\cdot)$, we have
	$$x_j({t},\signal{u}{j},\signal{d}{j},x_j(0))\leq
	x_j({t},u'_j(\cdot),\signal{d}{j},x_j(0)).$$ Similarly, if
	$\signal{d}{j}\leq d'_{j}(\cdot)$, we have
	$$x_j({t},\signal{u}{j},\signal{d}{j},x_j(0))\leq
	x_j({t},\signal{u}{j},d'_{j}(\cdot),x_j(0)).$$ Furthermore, if $x_j(0)\leq
	x'_j(0)$, we have $$x_j({t},\signal{u}{j}, \signal{d}{j},x_j(0))\leq
	x_j({t},\signal{u}{j},\signal{d}{j}, {x}'_j(0)).$$ The same relations hold for
	the state of uncontrolled vehicle $\gamma$.
\end{assumption} 
\begin{assumption}\label{assumption:y_nondecreasing_t}
	For all $i\in\allset$, $\dot{y}_i(t) \geq 0$ for all $t$. Thus, the outputs $y_j({t},\signal{u}{j},\signal{d}{j}, x_j(0))$ and $y_\gamma(t,\signal{w}{\gamma},\signal{d}{\gamma},x_\gamma(0))$ are 	non-decreasing in $t$.
\end{assumption}
\begin{assumption}\label{assumption:path-connected}
	For all $j\in\mathcal{C}$, ${x}_j(t,\signal{u}{j},\signal{d}{j},{x}_j(0))$ is continuously dependent on $\signal{u}{j}\in \mathcal{U}_j$, and the input signal space $\mathcal{U}_j$ is path-connected.
\end{assumption}
Let $\mathbf{x}(t,\aggsignal{u},\aggsignal{w},\aggsignal{d},\mathbf{x}(0))$ denote the aggregate state with $\aggsignal{u}\in\mathcal{U}, \aggsignal{w}\in\mathcal{W}$, and $\aggsignal{d}\in\mathcal{D}$. This can also be written as $\mathbf{x}(t)$ if the other arguments are not important.


Notice that the vehicle dynamics \eqref{equation:generalmodel} are subject to uncertainty originated from disturbances $d_j$ and $d_\gamma$ and an unknown driver-input $w_\gamma$. In addition to these, we consider a measurement noise
$\delta_i:=(\delta_{y,i},\delta_{v,i})$ for $i\in\allset$, where $\delta_{y,i}\in \mathbb{R}$ and $\delta_{v,i}\in \mathbb{R}$ are noises on the position measurement $y_{m,i}(t)$ and the speed measurement $v_{m,i}(t)$ at some time $t$, respectively. Then, the actual state $x_i(t)=(y_i(t), v_i(t))$ satisfies the following equation:
\begin{align}\label{equation:measurement}
y_i(t) = y_{m,i}(t)+\delta_{y,i}, && v_i(t)=v_{m,i}(t)+\delta_{v,i}.
\end{align}
Let $x_{m,i}(t)=(y_{m,i}(t), v_{m,i}(t))$ denote the state measurement. Then, \eqref{equation:measurement} can be rewritten as $x_i(t) = x_{m,i}(t)+\delta_i$. We make an assumption as follows.
\begin{assumption}\label{assumption:noise_bounded}
	The measurement noise is bounded, that is, $\delta_{i}\in[\delta_{i,min},
	\delta_{i,max}]$.
\end{assumption} 
The aggregate state measurement is denoted by
$\mathbf{x}_m(t)$, and then, $\mathbf{x}(t) = \mathbf{x}_{m}(t) + \delta$, where $\delta\in \Delta:=[\delta_{min}, \delta_{max}]$ is the aggregated measurement noise.

\subsection{The state estimation}
We define a set of states
$[\lowx(t),\upperx(t)]$, called the state estimation, that provides a lower and upper bound of the exact state, that is, $\mathbf{x}(t)\in[\lowx(t),\upperx(t)]$ for
all $t$. At $t=0$, the state estimation is defined as $\lowx(0) = \mathbf{x}_m(0)+\delta_{min}$ and $\upperx(0) = \mathbf{x}_m(0)+\delta_{max}$, so that $\mathbf{x}(0)\in[\lowx(0), \upperx(0)]$ because of \eqref{equation:measurement} and Assumption~\ref{assumption:noise_bounded}. Given the intial state estimation $[\lowx(0),
\upperx(0)]$ and an input signal $\aggsignal{u}$, the state estimation $[\lowx(t,\aggsignal{u},\lowx(0)), \upperx(t,
\aggsignal{u},\upperx(0))]$ at time $t$ is defined as follows: for $j\in\mathcal{C}$,
\begin{align}
\begin{split}\label{equation:controlledestimate}
&x^a_j(t,\signal{u}{j},x_j^a(0)) := \min_{	\begin{subarray}
	\centering
	\signal{d}{j}\in \mathcal{D}_j, x_j(0)\\\in[x^a_j(0), x^b_j(0)]
	\end{subarray}}
x_j(t,\signal{u}{j},\signal{d}{j},x_j(0)),\\
&x^b_j(t,\signal{u}{j},x_j^b(0)):=\max_{\begin{subarray}
\centering
\signal{d}{j}\in \mathcal{D}_j, x_j(0)\\\in[x^a_j(0), x^b_j(0)]
\end{subarray}}
x_j(t,\signal{u}{j},\signal{d}{j}, x_j(0)).
\end{split}
\end{align}
For $\gamma\in\uncontrolled,$
\begin{align}
\begin{split}\label{equation:uncontrolledestimate}
& x^a_\gamma(t) := \min_{	\begin{subarray}
	\centering
	\signal{w}{\gamma}\in \mathcal{W}_\gamma,\signal{d}{\gamma}\in \mathcal{D}_\gamma,\\
	x_\gamma(0)\in[x^a_\gamma(0),x^b_\gamma(0)]
	\end{subarray}} x_\gamma(t,\signal{w}{\gamma},\signal{d}{\gamma},x_\gamma(0)),\\
& x^b_\gamma(t):=\max_{
	\begin{subarray}
	\centering
	\signal{w}{\gamma}\in \mathcal{W}_\gamma,\signal{d}{\gamma}\in \mathcal{D}_\gamma,\\
	x_\gamma(0)\in[x^a_\gamma(0),x^b_\gamma(0)]
	\end{subarray}}
x_\gamma(t,\signal{w}{\gamma},\signal{d}{\gamma}, x_\gamma(0)).
\end{split}
\end{align}
The state estimation is denoted by $[\lowx(t,\aggsignal{u},\lowx(0)), \upperx(t,
\aggsignal{u},\upperx(0))]$ or $[\lowx(t),\upperx(t)]$ if the omitted arguments are not necessary.

By these definitions, the state estimation guarantees that given an initial state estimation $[\lowx(0),\upperx(0)]$ and an input signal $\aggsignal{u}$, 
\begin{align}\label{equation:stateestimation_guarantee}
\begin{split}
\mathbf{x}(t,\aggsignal{u},&\aggsignal{w},\aggsignal{d},\mathbf{x}(0))\\
&\in
[\lowx(t,\aggsignal{u},\lowx(0)), \upperx(t,\aggsignal{u},\upperx(0))],
\end{split}
\end{align}
for all $t$ for any $\aggsignal{w}\in\mathcal{W}$ and $\aggsignal{d}\in\mathcal{D}$ for any $\mathbf{x}(0)\in [\lowx(0),\upperx(0)]$. 


Notice that $x^a_j(t,\signal{u}{j},x^a_j(0))$ and $x^b_j(t,\signal{u}{j},x_j^b(0))$ satisfy the inequalities in Assumption~\ref{assumption:x_orderpreserving}. We can define the position estimation, denoted by $[\mathbf{y}^a(t,\aggsignal{u},\lowx(0)), \mathbf{y}^b(t,\aggsignal{u},\upperx(0)))]$ by letting $x^a_i(t)=(y^a_i(t), \cdot)$ and $x^b_i(t)=(y^b_i(t),\cdot)$ for all $i\in\allset$. Throughout this paper, we use \textquoteleft$\cdot$' in a vector if specifying the entry is not important. The position estimation also inherits the order-preserving property from $\mathbf{y}(t)$. Later in
Section~\ref{section:supervisor}, this estimation can be updated using the state measurement.

\begin{figure*}[t!]
\centering
\includegraphics[width = 0.7\linewidth]{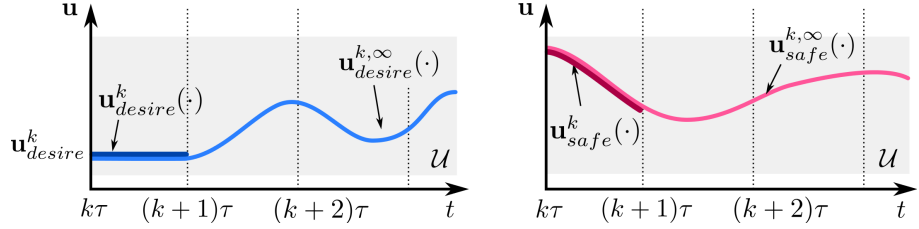}
\caption{$\aggdesire{k}$ is an artificially constructed signal that equals to $\mathbf{u}^k_{desire}$ for time $[k\tau,(k+1)\tau)$ as in \eqref{eq:desired_input}. $\aggsafe{k,\infty}$ is an safe input signal that makes the system avoid the Bad set as in \eqref{eq:safe_input}.}
\label{figure:input_signal}
\end{figure*}
\begin{figure*}[t!]
\centering
\includegraphics[width = 0.95\linewidth]{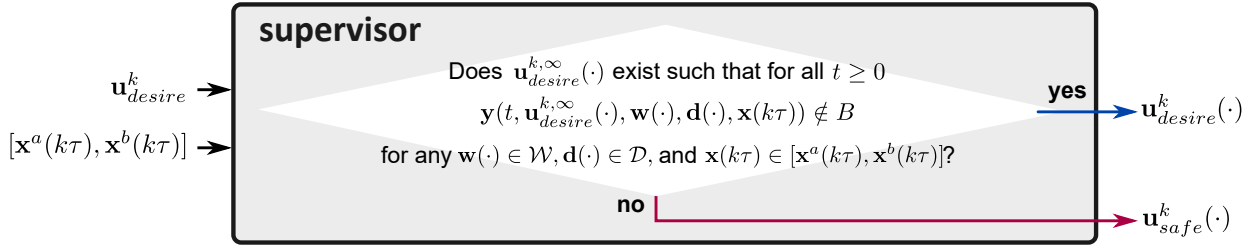}
\caption{Illustration of the supervisor. At time $k\tau$, the supervisor takes the desired input of controlled vehicles' drivers $\mathbf{u}^{k}_{desire}$ and the state estimation $[\lowx(k\tau),\upperx(k\tau)]$. Depending on the existence of $\aggdesire{k}$ that prevents all the vehicles from entering the Bad set for all $t$ for any $\aggsignal{w}\in\mathcal{W},\aggsignal{d}\in\mathcal{D},$ and $\mathbf{x}(k\tau)\in[\lowx(k\tau),\upperx(k\tau)]$, the supervisor either allows the drivers to travel with their desired input $\mathbf{u}_{desire}^{k}(\cdot)$ or overrides them with the safe input $\mathbf{u}_{safe}^{k}(\cdot)$.}
\label{figure:supervisor_design_problem}
\end{figure*}

\section{Problem Statement}\label{section:problemstatement}

Let us define an intersection collision. Since an intersection is modeled as a single conflict area, we consider a
collision occurs if at least two vehicles are
simultaneously inside the intersection. The output configurations corresponding to collisions belong to the \textit{Bad set}, denoted by $B\subset Y$. The Bad set is defined as follows:
\begin{align*}
B:= \{&\mathbf{y}\in Y: y_i\in (\alpha_i, \beta_i) \text{ and } y_j\in
(\alpha_j, \beta_j)\\
&\text{ for some } i\ne j \text{ such that }i\in \allset\text{ and } j\in
\mathcal{C}\}.
\end{align*}
Throughout this paper, we assume that uncontrolled vehicles do not crash among themselves. This assumption enables us to focus on preventing collisions in which at least one controlled vehicle is involved.

A supervisor runs in discrete time with a time step $\tau$. At time $k\tau$ where $k$ is a nonnegative integer, a desired input $\mathbf{u}_{desire}^k\in \mathbf{U}$ and a state $\mathbf{x}_m(k\tau)\in \mathbf{X}$ are measured. Then, a state estimation $[\lowx(k\tau), \upperx(k\tau)]$ is updated to satisfy $\mathbf{x}(k\tau)\in[\lowx(k\tau),\upperx(k\tau)]\subseteq [\mathbf{x}_m(k\tau)+\delta_{min}, \mathbf{x}_m(k\tau)+\delta_{max}]$. The desired input $\mathbf{u}_{desire}^k$ is a vector of current inputs of controlled vehicles' drivers at time $k\tau$. As illustrated in Figure~\ref{figure:input_signal}, we define desired input signals $\mathbf{u}_{desire}^k(\cdot)\in\mathcal{U}$ on time $[k\tau, (k+1)\tau)$ and $\aggdesire{k}\in \mathcal{U}$ on time $[k\tau, \infty)$ such that \begin{equation}\label{eq:desired_input}
\mathbf{u}_{desire}^k(t)=\mathbf{u}^{k,\infty}_{desire}(t)=\mathbf{u}_{desire}^k
~\text{for}~t\in [k\tau, (k+1)\tau).
\end{equation}   Also, we define a safe input signal $\aggsafe{k,\infty}$ on time $[k\tau,\infty)$ such that
\begin{align}
\begin{split}\label{eq:safe_input}
\forall&\aggsignal{w}\in\mathcal{W}, \aggsignal{d}\in\mathcal{D},\mathbf{x}(k\tau)\in[\lowx(k\tau),\upperx(k\tau)],\\
&\mathbf{y}(t,\aggsafe{k,\infty},\aggsignal{w},\aggsignal{d},\mathbf{x}(k\tau))\notin B~ \text{for all}~t.
\end{split}
\end{align}  Let $\aggsafe{k}$ be $\aggsafe{k,\infty}$ restricted to time $[k\tau, (k+1)\tau)$.  

The supervisor $s([\lowx(k\tau), \upperx(k\tau)],\mathbf{u}_{desire}^k)$ is designed according to the following problem, which is illustrated in Figure~\ref{figure:supervisor_design_problem}.

\begin{problem}[(Supervisor-design)]\label{problem:supervisor}
	Design a supervisor $s([\lowx(k\tau),\upperx(k\tau)],\mathbf{u}_{desire}^k)$ such that for any $\aggsignal{w}, \aggsignal{d},$ and $\mathbf{x}(k\tau)\in[\lowx(k\tau),\upperx(k\tau)]$, it returns
	$$	\begin{cases}
	\mathbf{u}_{desire}^k(\cdot) & \text{if}~\exists \aggdesire{k}:\text{
		for all}~ t\geq 0\\ 
	&\mathbf{y}(t,\aggdesire{k},\aggsignal{w},\aggsignal{d},\mathbf{x}(k\tau))\notin B\\
	
	\aggsafe{k} & \text{otherwise,}
	\end{cases}
	$$
	and it is non-blocking, that is, if $s([\lowx ((k-1)\tau),\upperx
	((k-1)\tau)],\mathbf{u}_{desire}^{k-1})\neq \emptyset$, then for any $\mathbf{u}_{desire}^{k}\in\mathbf{U},$ we have $ s([\lowx
	(k\tau),\upperx (k\tau)],\mathbf{u}_{desire}^{k})\neq \emptyset.$
\end{problem}

The supervisor in Problem~\ref{problem:supervisor} is least restrictive in the sense that overrides are activated only when the desired input of controlled vehicles' drivers would lead to collisions at some future times. 
To distinguish between safe and unsafe inputs, the supervisor needs to verify that the state reached using the desired input is compatible with a safe evolution of the system. This introduces the following problem.

\begin{problem}[(Verification)]\label{problem:verification}
	Given an initial state estimation $[\lowx(0), \upperx(0)]$, determine if there
	exists an input signal $\aggsignal{u}$ that guarantees
	$\mathbf{y}(t,\aggsignal{u},\aggsignal{w},\aggsignal{d},\mathbf{x}(0))\notin B$ for all $t\geq 0$
	for any $\aggsignal{w}\in\mathcal{W}, \aggsignal{d}\in\mathcal{D},$ and
	$\mathbf{x}(0)\in[\lowx(0),\upperx(0)]$.
\end{problem}

\begin{figure*}[t!]
	\centering
	\includegraphics[width=0.7\linewidth,keepaspectratio=true,angle=0]{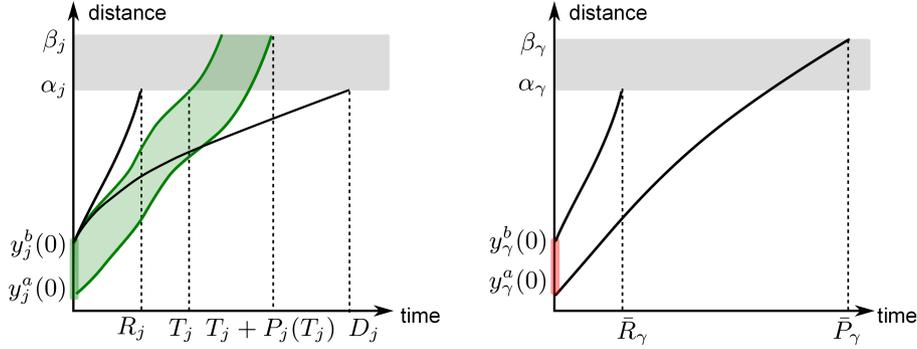}
	\caption{Scheduling parameters in
Definition~\ref{definition:schedulingparmeters}. Referring to Figure~\ref{figure:generalcars}, $(\alpha_j,\beta_j)$ is the location of an intersection along the route of vehicle~$j$.}
	\label{figure:parameters}
\end{figure*}
By solving Problem~\ref{problem:verification}, the supervisor determines if an override is necessary. If the desired input leads to a set of states at which  Problem~\ref{problem:verification} returns \textit{no}, the supervisor overrides controlled vehicles with a safe input signal. The two above problems are solved exactly in the next section, and approximately in Section~\ref{section:approxsol}.


\section{Exact solutions}\label{section:exactsol}

In this section, we provide the solutions of the verification problem (Problem~\ref{problem:verification}) and the supervisor-design problem (Problem~\ref{problem:supervisor}). To solve Problem~\ref{problem:verification}, we formulate an Inserted Idle-Time scheduling problem, which is solved straightforwardly, and prove that this problem is equivalent to Problem~\ref{problem:verification}. We then propose a solution to Problem~\ref{problem:supervisor}.

\subsection{Inserted Idle-Time scheduling problem}

Problem~\ref{problem:verification} can be translated into a scheduling problem, considering the intersection as a resource
that all vehicles must share. The schedule represents a sequence of the times at which each controlled vehicle enters an intersection, and the idle-times represent the sets of times during which an intersection is potentially occupied by uncontrolled vehicles. This analogy is characterized mathematically using the concept of decision problem equivalence.


\begin{definition}(\label{definition:equivalence}\cite{cormen_introduction_2009})
Consider two decision problems $A$ and $B$. The input to a particular problem is
called an \textbf{instance} of that problem. Then, $A$ is reducible to $B$ if
there is a procedure that transforms any instance $\alpha$ of $A$ into some
instance $\beta$ of $B$ in polynomial time, and the answer for $\alpha$ is
\textquotedblleft yes", denoted by $\alpha\in A$, if and only if the answer for
$\beta$ is \textquotedblleft yes", denoted by $\beta\in B$. We say $A$ is
\textbf{equivalent} to $B$ if and only if $A$ is reducible to $B$ and $B$ is
reducible to $A$.
\end{definition}
An instance $I$ of Problem~\ref{problem:verification} is described by $(
[\lowx(0),\upperx(0)],S)$ where $S:=(\mathbf{F,X,Y},\mathcal{U},\mathcal{W},
\mathcal{D},\mathcal{C},\bar{\mathcal{C}},\alpha_1,\ldots,\alpha_n,\beta_1,\ldots,\beta_n)$.

To formulate the IIT scheduling problem, we define scheduling parameters.

\begin{definition}\label{definition:schedulingparmeters}
Given an initial state estimation $[\lowx(0), \upperx(0)]$, release times $R_j$,
deadlines $D_j$, and process times $P_j(T_j)$ are defined for controlled
vehicles. For $j\in \mathcal{C}$, if $y_j^b
(0)< \alpha_j$,
\begin{align*}
& R_j := \min_{\signal{u}{j}\in \mathcal{U}_j} \{ t \geq 0: y^b_{j}
(t,\signal{u}{j},x^b_j(0)) = \alpha_j \},\\
& D_j := \max_{\signal{u}{j}\in \mathcal{U}_j} \{ t\geq 0:
y^b_{j}(t,\signal{u}{j},x^b_j(0))=\alpha_j \}.
\end{align*}
Given a real number $T_j> 0$,
\begin{align*}
P_j(T_j):= &\min_{\signal{u}{j}\in \mathcal{U}_j}\{ t-T_j\geq 0:\\
&y^a_{i}(t,\signal{u}{j},x^a_j(0))=\beta_j\\
&\hspace{0.3 in} \text{with constraint}~ y^b_j(T_j,\signal{u}{j},x^b_j(0)) = \alpha_j\}.
\end{align*}
If $y^b_{j}(0) \geq \alpha_j$, then $R_j=0, D_j=0$, and
$P_j(T_j)=\min_{\signal{u}{j}\in \mathcal{U}_j}\{ t:
y^a_{j}(t,\signal{u}{j},x^a_j(0))=\beta_j\}$. If $y^a_{j}(0) \geq \beta_j$, then
set $R_j=0, D_j=0$, and $P_j(T_j)=0$. If the constraint is not satisfied, set
$P_j(T_j)=\infty$.

Idle-times $(\bar{R}_\gamma, \bar{P}_\gamma)$ are defined for uncontrolled
vehicles. For $\gamma\in \bar{\mathcal{C}}$, if $y^b_\gamma(0) <
\alpha_{\gamma}$,
\begin{align*}
& \bar{R}_{\gamma}:= \{t\geq 0: y^b_\gamma(t)=\alpha_\gamma\},\\
& \bar{P}_{\gamma}:=\{t\geq 0: y^a_\gamma(t)=\beta_\gamma\}.
\end{align*}
If $y^b_\gamma(0)\geq \alpha_\gamma$, set $\bar{R}_\gamma=0$ and
$\bar{P}_\gamma=\{t: y^a_\gamma(t)=\beta_\gamma\}$. If $y^a_\gamma(0) \geq
\beta_\gamma$, set $\bar{R}_\gamma=0$ and $\bar{P}_\gamma=0$. 
\end{definition}

In this definition, release time $R_j$ is the earliest that controlled
vehicle $j$ can enter an intersection while deadline $D_j$ is the
latest. Given that controlled vehicle $j$ enters the intersection no earlier than time $T_j$, process time $P_j(T_j)$ is the earliest that it can cross the intersection. Uncontrolled vehicle $\gamma$
enters and exits the intersection within the idle-time $(\bar{R}_\gamma,
\bar{P}_\gamma)$ regardless of its driver-input and disturbance signals. These parameters are illustrated in Figure~\ref{figure:parameters}.

The Inserted Idle-time (IIT) scheduling problem is formulated as follows.
\begin{problem}[(IIT scheduling)]\label{problem:scheduling}
	Given an initial state estimation $[\lowx(0), \upperx(0)]$, determine whether
there exists a schedule $\mathbf{T}=(T_1,\ldots, T_{n_c})\in
\mathbb{R}^{n_c}_{+}$ such that	for all $j\in \mathcal{C}$, 
	\begin{equation}
	\label{condition:boundedinput}
	R_j\leq T_j\leq D_j,
	\end{equation}
	for all $i\ne j \in \mathcal{C}$,
	\begin{equation}
	\label{condition:controlled}
	(T_i, T_i+P_i(T_i))\cap(T_j,T_j+P_j(T_j)) = \emptyset,
	\end{equation}
	for all $j\in \mathcal{C}$ and $\gamma \in \uncontrolled$,
	\begin{equation}
	(T_j, T_j+P_j(T_j)) \cap  (\bar{R}_\gamma,\bar{P}_\gamma) = \emptyset.
\label{condition:uncontrolled}
	\end{equation}
\end{problem}

Notice that $T_j$ and $P_j(T_j)$ are defined such that for some $\signal{u}{j}$, $y_j^b(T_j,\signal{u}{j},x_j^b(0))=\alpha_j$ and $y_j^a(T_j+P_j(T_j),\signal{u}{j},x_j^a(0))=\beta_j$. Condition~\eqref{condition:boundedinput} represents the constraint induced on the schedule by the bounded input signals. By
condition~\eqref{condition:controlled}, the times during which controlled
vehicles $i$ and $j$ occupy the intersection for any disturbance do not
overlap. This implies that a collision between controlled vehicles $i$ and
$j$ is averted. Similarly, condition~\eqref{condition:uncontrolled} implies that vehicle $j$ does not occupy the intersection during the
idle-time, thereby preventing a collision between controlled vehicle $j$ and uncontrolled vehicle $\gamma$. Thus, a schedule satisfying the above conditions is related to an input
signal that can prevent any future collision. This is the essence of the proof of the following theorem.

\begin{theorem}\label{theorem:equivalence}
	Problem \ref{problem:verification} and Problem \ref{problem:scheduling} are
equivalent.
\end{theorem}

\begin{proof}
By Definition~\ref{definition:equivalence}, we need to show two things: Problem~\ref{problem:verification} is reducible to
Problem~\ref{problem:scheduling}, and Problem~\ref{problem:scheduling}
is reducible to Problem~\ref{problem:verification}. Notice that an instance $I$ of Problem~\ref{problem:scheduling} is also described by $([\lowx(0),\upperx(0)],S)$, which is the same as an instance of Problem~\ref{problem:verification}. Thus, the transformation between the instances of these problems takes constant time. Now, the following relation is left to prove the equivalence.
$$I\in \text{Problem~\ref{problem:verification}} \Leftrightarrow I\in \text{
Problem~\ref{problem:scheduling}}.$$

($\Rightarrow$) Given an initial state estimation $[\lowx(0), \upperx(0)]$, there is an input signal $\aggsignal{\tilde{u}}\in
\mathcal{U}$ such that $\mathbf{y}(t,\aggsignal{\tilde{u}},\aggsignal{w},
\aggsignal{d},\mathbf{x}(0))\notin B$ for all $t$ for any
$\aggsignal{w}\in\mathcal{W}$ and $\aggsignal{d}\in \mathcal{D}$ for any
$\mathbf{x}(0)\in[\lowx(0),\upperx(0)]$.


Using $\signal{\tilde{u}}{j}$ for $j\in\mathcal{C}$, which denotes the $j$-th
element of $\aggsignal{\tilde{u}}$, let $\tilde{T}_j$ be the time at which
$y^b_j(t,\signal{\tilde{u}}{j},x^b_j(0))=\alpha_j$ if $y^b_j(0)< \alpha_j$. If $y^b_j(0)\geq \alpha_j$, set
$\tilde{T}_j=0$. This
satisfies the constraint of $P_j(\tilde{T}_j)$ in
Definition~\ref{definition:schedulingparmeters}. Let $\tilde{T}_j+\tilde{P}_j(\tilde{T}_j)$ be the time at which
$y^a_j(t,\signal{\tilde{u}}{j},x^a_j(0))=\beta_j$. Set
$\tilde{P}_j(\tilde{T}_j)=0$ if $y^a_j(0)\geq \beta_j$. 

By the definitions of $R_j$ and $D_j$, condition~\eqref{condition:boundedinput} is satisfied. Suppose
without loss of generality, $\tilde{T}_i \leq \tilde{T}_j$. At $t=\tilde{T}_j$, we have $y_j^b(t,\signal{\tilde{u}}{j},x^b_j(0))= \alpha_j$. Since $\signal{\tilde{u}}{i}$ and $\signal{\tilde{u}}{j}$ guarantee that at most one vehicle is inside an intersection, we must have $y_i^a(t,\signal{\tilde{u}}{i},x_i^a(0))\geq \beta_i$. Because $y^a_j(t)$ is non-decreasing in time, $\tilde{T}_i+\tilde{P}_i(\tilde{T}_i)\leq t$. By the definition of $P_i(\tilde{T}_i)$, we have
$\tilde{T}_i+P_i(\tilde{T}_i)\leq \tilde{T}_i+\tilde{P}_i(\tilde{T}_i)\leq t=\tilde{T}_j$,
which concludes $(\tilde{T}_i,
\tilde{T}_i+P_i(\tilde{T}_i))\cap(\tilde{T}_j,\tilde{T}_j+P_j(\tilde{T}_j)) =
\emptyset$ (condition~\eqref{condition:controlled}). 

In order to avoid the Bad set, vehicle $\gamma\in\uncontrolled$ is not inside an intersection during $[\tilde{T}_j,\tilde{T}_j+\tilde{P}(\tilde{T}_j)]$  for any $\signal{w}{\gamma}\in\mathcal{W}_\gamma, \signal{d}{\gamma}\in\mathcal{D}_{\gamma},$ and $x_\gamma(0)\in [x^a_\gamma(0), x^b_\gamma(0)]$. Since $y_\gamma^b(\bar{R}_\gamma)=\alpha_\gamma$ and $y_\gamma^a(\bar{P}_\gamma)=\beta_\gamma$ by Definition~\ref{definition:schedulingparmeters}, we have $(\tilde{T}_j,
\tilde{T}_j+\tilde{P}_j(\tilde{T}_j))\cap(\bar{R}_\gamma,\bar{P}_\gamma) =
\emptyset$. By the definition of
$P_j(\tilde{T}_j)$, we have $(\tilde{T}_j, \tilde{T}_j+P_j(\tilde{T}_j))\subset (\tilde{T}_j, \tilde{T}_j+\tilde{P}_j(\tilde{T}_j))$. Thus, $(\tilde{T}_j, \tilde{T}_j+P_j(\tilde{T}_j)) \cap 
(\bar{R}_\gamma,\bar{P}_\gamma) = \emptyset$.
(condition~\eqref{condition:uncontrolled}). 

This proves that there exists the schedule $\tilde{\mathbf{T}}$ that satisfies
conditions~\eqref{condition:boundedinput}, \eqref{condition:controlled}, and
\eqref{condition:uncontrolled}.

\medskip\noindent($\Leftarrow$) Given an initial state estimation $[\lowx(0),
\upperx(0)]$, there exists a schedule $\tilde{\mathbf{T}}\in \mathbb{R}^{n_c}$
that satisfies the conditions of Problem~\ref{problem:scheduling}. 

For $j\in\mathcal{C}$, define $\signal{\tilde{u}}{j}$ such that $y^b_j(\tilde{T}_j,\signal{\tilde{u}}{j},x^b_j(0))=\alpha_j$ if $y^b_j(0)<\alpha_j$. If $y^b_j(0)\geq \alpha_j$ and $y^a_j(0)\leq \beta_j$, define $\signal{\tilde{u}}{j}$ such that $y^a_j(\tilde{T}_j+P_j(\tilde{T}_j),\signal{\tilde{u}}{j},x^a_j(0))=\beta_j$. If $y^a_j(0)>\beta_j$, we do not consider vehicle $j$ because it has already crossed the intersection. Since $y^a_j(t,\signal{u}{j},x^a_j(0))$ depends continuously on $\signal{u}{j}\in\mathcal{U}_j$ and $\mathcal{U}_j$ is path connected by Assumption~\ref{assumption:path-connected}, condition~\eqref{condition:boundedinput} implies that such an input signal exists, i.e., $\signal{\tilde{u}}{j}\in\mathcal{U}_j$.

Consider $\signal{\tilde{u}}{i}$ and $\signal{\tilde{u}}{j}$ for $i\ne j\in\mathcal{C}$ where $\tilde{T}_i \leq \tilde{T}_j$. Condition~\eqref{condition:controlled} says
$\tilde{T}_i+P_i(\tilde{T}_i)\leq \tilde{T}_j$. At time
$t\in[\tilde{T}_i+P_i(\tilde{T}_i),\tilde{T}_j]$, we have
$y^a_i(t,\signal{\tilde{u}}{i},x^a_i(0))\geq \beta_i$ and $y^b_j(t,
\signal{\tilde{u}}{j},x^b_j(0))\leq \alpha_j$. Thus, for any disturbance and initial condition, vehicle $j$ enters the intersection after vehicle $i$ leaves it.

For $j\in\mathcal{C}$ and $\gamma\in\uncontrolled$,
condition~\eqref{condition:uncontrolled} implies either $\bar{P}_\gamma\leq T_j$
or $\tilde{T}_j+P_j(\tilde{T}_j) \leq \bar{R}_\gamma$. In the first case, at any
time $t\in[\bar{P}_\gamma, T_j]$, we have $y^a_\gamma(t)\geq \beta_\gamma$ and 
$y^b_j(t,\signal{\tilde{u}}{j},x^b_j(0))\leq \alpha_j$ so that for any $\signal{w}{\gamma}\in\mathcal{W}_\gamma, \signal{d}{\gamma}\in\mathcal{D}_\gamma, x_\gamma(0)\in[x_\gamma^a(0),x_\gamma^b(0)]$ and for any $\signal{d}{j}\in\mathcal{D}_j, x_j(0)\in[x_j^a(0), x_j^b(0)]$, after vehicle $\gamma$
exits the intersection, vehicle $j$ enters it. In the second case, at any time
$t\in[\tilde{T}_j+P_j(\tilde{T}_j),\bar{R}_\gamma]$, we have
$y^a_j(t,\signal{\tilde{u}}{j},x^a_j(0))\geq\beta_j$ and $y^b_\gamma(t)\leq
\alpha_\gamma$. This means after vehicle $j$ exits the intersection for any $\signal{d}{j}\in\mathcal{D}_j$ and $x_j(0)\in[x_j^a(0), x_j^b(0)]$, vehicle $\gamma$ enters it for any $\signal{w}{\gamma}\in\mathcal{W}_\gamma, \signal{d}{\gamma}\in\mathcal{D}_\gamma,$ and $x_\gamma(0)\in[x_\gamma^a(0),x_\gamma^b(0)]$. 

Therefore, the schedule $\tilde{\mathbf{T}}$ ensures that there exists an input
signal $\aggsignal{\tilde{u}}$ which prevents entering the Bad set. \qed
\end{proof}

We now provide an algorithm that solves Problem~\ref{problem:scheduling}, which, in turn, solves Problem~\ref{problem:verification} by Theorem~\ref{theorem:equivalence}. This algorithm contains two procedures. The first procedure, \procedureS, generates a schedule $\mathbf{T}$ given a sequence $\pi$ such that $T_{\pi_i}\leq T_{\pi_j}$ if $i<j$, and evaluates whether $\mathbf{T}$ satisfies conditions~\eqref{condition:boundedinput}-\eqref{condition:uncontrolled}. Here, $\pi_i$ denotes the $i$-th entry of $\pi$ and $T_{\pi_i}$ the $\pi_i$-entry of $\mathbf{T}$. The second procedure, \procedureES, inspects all possible sequences until a sequence with a feasible schedule is found. If a feasible schedule is found, the answer of Problem~\ref{problem:verification} is \textit{yes}, and otherwise, the answer is \textit{no}.

Algorithm~\ref{algorithm:verification} focuses on vehicles before an intersection. The set of such vehicles is denoted by $\mathcal{M}:=\{j\in \mathcal{C}:y_j^b(0)<\alpha_j\}$. For simplicity, let $y^b_j(0)<\alpha_j$ for $j\in \{1,\ldots,|\mathcal{M}|\}$ and $y^b_j(0) \geq \alpha_j$ for $j\in\{|\mathcal{M}|+1,\ldots,n_c\}$. Let $\mathcal{P}$ be a set of all permutations of $\mathcal{M}$. Its
cardinality is $(|\mathcal{M}|)!=1\times 2\times \ldots\times |\mathcal{M}|$. Let $\pi\in
\mathbb{R}^{|\mathcal{M}|}$ be a vector in $\mathcal{P}$. Let $\bar{\pi}\in\mathbb{R}^{n_{\bar{c}}}$ be a vector of $\gamma\in \bar{\mathcal{C}}$
in an increasing order of $\bar{R}_\gamma$, that is, $\bar{R}_{\bar{\pi}_{j}}\leq \bar{R}_{\bar{\pi}_{j'}}$ for $j\leq j'$. Let $\pi_{0,j}$ and $\bar{\pi}_j$ denote the
$j$-th entry of $\pi_0$ and $\bar{\pi}$, respectively.
\begin{figure*}[t!]
	\centering
	\includegraphics[width=\linewidth]{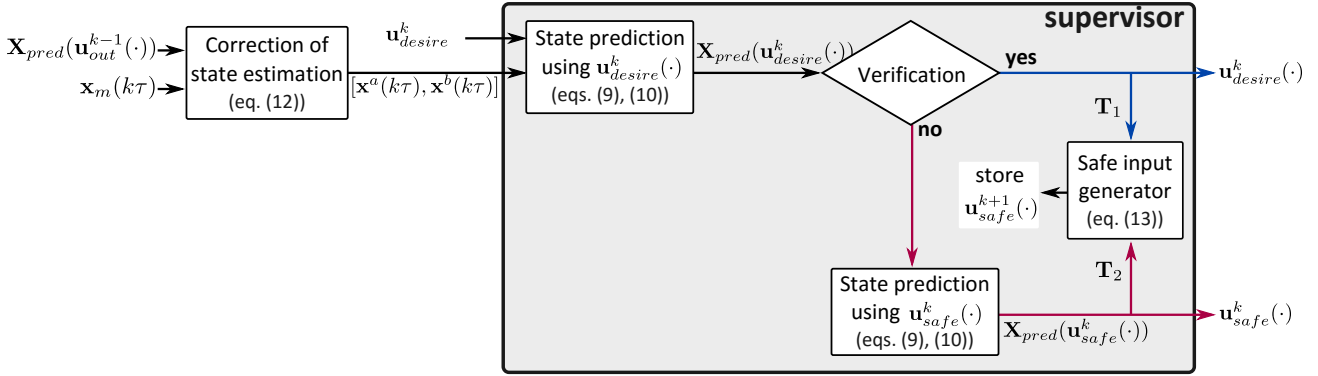}
	\caption{Schematic of the supervisor at time $k\tau$. Using the state measurement $\mathbf{x}_m(k\tau)$ and the previous step's state prediction $\mathbf{X}_{pred}(\aggout{k-1})$, the state estimation is corrected to $[\lowx(k\tau),\upperx(k\tau)]$. The supervisor takes $[\lowx(k\tau),\upperx(k\tau)]$ and the current desired input $\mathbf{u}_{desire}^k$ and returns either $\mathbf{u}_{desire}^k(\cdot)$ or $\aggsafe{k}$ depending on the answer of the verification problem. When the answer is \textit{no}, the state prediction $\mathbf{X}_{pred}(\aggsafe{k})$ guarantees the existence of a feasible schedule $\mathbf{T}_2$, thereby ensuring the non-blocking property of the supervisor. At each step, a safe input signal $\aggsafe{k+1}$ is generated using a feasible schedule and stored for a possible use at the next step.}
	\label{figure:supervisor}
\end{figure*}
\begin{algorithm}[H]
	\begin{algorithmic}[1]
		\Procedure{\procedureS}{$\pi_0,[\lowx(0), \upperx(0)],S$}
		\ForAll{$j\in\mathcal{C}$} {calculate $R_j,D_j$}\label{scheduling:RjDj}
		\EndFor
		\ForAll{$\gamma\in\uncontrolled$} {calculate
			$\bar{R}_{\gamma},\bar{P}_{\gamma}$}\label{scheduling:RbarPbar}
		\EndFor
		\State $\mathcal{M} = \{j\in\mathcal{C}: y^b_j(0)<\alpha_j\}$ and $\bar{\mathcal{M}}=\mathcal{C}\setminus \mathcal{M}$ \label{scheduling:M_and_Mbar}
		\State $T_j = 0~\forall~j\in\bar{\mathcal{M}}$ and $P_{max} = \max_{j\in \bar{\mathcal{M}}} P_j(T_j)$ \label{scheduling:Pmax}
		\State $\pi = (\pi_{0,j}: \pi_{0,j}\in \mathcal{M}~\text{for}~j=1~\text{to}~|\pi_0|)$\label{scheduling:pi_ne_M}
		\For{$i= 1$ to $|\mathcal{M}|$}
		\label{scheduling:bigforloop}
		\If{$i=1$} $T_{\pi_1}= \max(R_{\pi_1},P_{max})$\label{scheduling:i=1} 
		\ElsIf{$i\geq 2$} 
		\State $T_{\pi_i}=
		\max(R_{\pi_i},T_{\pi_{i-1}}+P_{\pi_{i-1}}(T_{\pi_{i-1}}))$\label{scheduling:con1and2}
		\EndIf
		
		\For{$j= 1$ to $n_{\bar{c}}$}\label{scheduling:innerforloop}
		\If{$T_{\pi_i} \geq \bar{R}_{\bar{\pi}_j}$} \label{scheduling:Tj>Rbar}
		\State $T_{\pi_i}=
		\max(T_{\pi_i},\bar{P}_{\bar{\pi}_j})$\label{scheduling:Tj>Rbar_then>Pbar}
		\ElsIf{$T_{\pi_{i}}+P_{\pi_i}(T_{\pi_i})>\bar{R}_{\bar{\pi}_j}$}\label{scheduling:Pj>Rbar}
		\State $T_{\pi_i}=
		\bar{P}_{\bar{\pi}_j}$\label{scheduling:Pj>Rbar_thenTj=Pbar}
		\EndIf
		\EndFor\label{scheduling:endinnerfor}
		\EndFor\label{scheduling:endi=2tom}
		\If{$T_j \leq D_j$ for all $j\in \mathcal{C}$} \textbf{return} $(\mathbf{T},yes)$\label{scheduling:checkT<D}
		\EndIf
		\State \textbf{return} $(\emptyset,no)$
		\EndProcedure\label{scheduling:lastline}
		
		\Procedure{\procedureES}{$[\lowx(0), \upperx(0)],S$}
		\If{$[\mathbf{y}^a(0), \mathbf{y}^b(0)]\cap B\ne \emptyset$}~\textbf{return} $(\emptyset,no)$\label{exact:already_in_B}
		\EndIf
		\State $\mathcal{M} = \{j\in\mathcal{C}: y^b_j(0)<\alpha_j\}$ and $\bar{\mathcal{M}}=\mathcal{C}\setminus \mathcal{M}$ \label{exact:M_and_Mbar}
		\If{$|\mathcal{M}| = 0$} \textbf{return} $(\emptyset,yes)$ \EndIf\label{exact:Mzero}
		\ForAll{$\pi\in\mathcal{P}$} \label{exact:permutation}
		\State $(\mathbf{T},ans)=\procedureS(\pi, [\lowx(0), \upperx(0)],S)$
		\If{$ans=yes$} \textbf{return} $(\mathbf{T},yes)$ 
		\EndIf
		\EndFor\label{exact:end_permutation}	
		\State \textbf{return} $(\emptyset,no)$
		\EndProcedure\label{exact:lastline}
	\end{algorithmic}
	\caption{Exact Solution of Problem~\ref{problem:scheduling}}
	\label{algorithm:verification}
\end{algorithm}

Procedure \procedureS~in Algorithm~\ref{algorithm:verification} works as
follows. In lines \ref{scheduling:RjDj}-\ref{scheduling:RbarPbar}, the scheduling
parameters in Definition~\ref{definition:schedulingparmeters} are computed given
an initial state estimation. In lines~\ref{scheduling:M_and_Mbar} and \ref{scheduling:Pmax}, if $y^b_j(0)\geq \alpha_j$, $T_j=0$ because $R_j=D_j=0$ and $P_{max}$ is the time at which all vehicles $j\in\bar{\mathcal{M}}$ exit the intersection. The input argument $\pi_0$ is the vector of the indexes of all vehicles controlled by the supervisor. The assignment at line~\ref{scheduling:pi_ne_M} extracts from $\pi_0$ the sub-vector $\pi$ of vehicles which have not yet reached the intersection. The schedule needs only be computed for the sub-vector $\pi$.


In a given sequence $\pi$, vehicle $\pi_{i-1}$ crosses
the intersection earlier than vehicle $\pi_{i}$. Given this sequence $\pi$, procedure \procedureS~finds the earliest possible schedule. 
In the \textbf{for} loop of
lines~\ref{scheduling:innerforloop}-\ref{scheduling:endinnerfor}, uncontrolled vehicle $\bar{\pi}_j$ is considered. If
line~\ref{scheduling:Tj>Rbar} is true, then
line~\ref{scheduling:Tj>Rbar_then>Pbar} ensures that $T_{\pi_i}\geq
\bar{P}_{\bar{\pi}_j}$ so that condition~\eqref{condition:uncontrolled} is
satisfied. Otherwise, if line~\ref{scheduling:Pj>Rbar} is true, meaning that
$T_{\pi_i}<\bar{R}_{\bar{\pi}_j}< T_{\pi_{i}}+P_{\pi_i}(T_{\pi_i})$, then
$T_{\pi_i}$ is delayed to $\bar{P}_{\bar{\pi}_j}$ in line~\ref{scheduling:Pj>Rbar_thenTj=Pbar} so that
condition~\eqref{condition:uncontrolled} becomes satisfied. The schedule $T_{\pi_i}$ takes $\bar{P}_{\bar{\pi}_j}$, which is the earliest possible value. Otherwise,
$T_{\pi_{i}}+P_{\pi_i}(T_{\pi_i})\leq \bar{R}_{\bar{\pi}_j}$ so that
condition~\eqref{condition:uncontrolled} is guaranteed. Lastly, in
line~\ref{scheduling:checkT<D}, the right inequality of
condition~\eqref{condition:boundedinput} is checked. If
line~\ref{scheduling:checkT<D} is true, this procedure returns a feasible schedule and answer \textit{yes}. If $T_j > D_j$ for some $j$, since the schedule $\mathbf{T}$ is constructed so that it is the earliest possible value, there cannot be another schedule $\tilde{\mathbf{T}}$ such that $\tilde{T}_j \leq D_j$. Thus, no feasible schedule can be found in this sequence $\pi$, and the procedure returns \textit{no}.

Procedure~\procedureES~in Algorithm~\ref{algorithm:verification} solves Problem~\ref{problem:scheduling} by inspecting all 
permutations in $\mathcal{P}$ until a feasible schedule is found as noted in lines~\ref{exact:permutation}-\ref{exact:end_permutation}. In line~\ref{exact:already_in_B}, if a given initial position estimation is already inside the Bad set, the procedure returns \textit{no}. In line~\ref{exact:Mzero}, if $y^b_j(0) \geq \alpha_j$ for all $j\in\mathcal{C}$, then it returns \textit{yes} since all controlled vehicles have entered the intersection. If no feasible schedule is found after evaluating all permutations in $\mathcal{P}$, an empty set and answer \textit{no} are returned in line~\ref{exact:lastline}.

The running time of procedure~\procedureES~is determined by the \textbf{for}
loops of lines~\ref{exact:permutation}-\ref{exact:end_permutation}.
For the worst case, all permutations in $\mathcal{P}$ must be evaluated. As the
number of controlled vehicles increases,
the computation time increases exponentially. Indeed, a scheduling problem is known to be NP-hard (\cite{colombo_efficient_2012}).
To avoid this computational complexity issue, we devise an approximate
algorithm to solve the IIT scheduling problem with the guarantee of a quantified approximation bound. This approximate algorithm is provided in Section~\ref{section:efficientverification}.

\subsection{Exact supervisor}\label{section:supervisor}
In this section, we solve Problem~\ref{problem:supervisor} by providing an algorithm implementing the exact supervisor. 


The schematic of a supervisor at time $k\tau$ is illustrated in Figure~\ref{figure:supervisor}. The state prediction computed at the previous step $\mathbf{X}_{pred}(\aggout{k-1})$ is known, where $\aggout{k-1}$ is the output that the supervisor returns at the previous step, that is, controlled vehicles traveled with $\aggout{k-1}$ for time $[(k-1)\tau, k\tau)$. Also, the state measurement $\mathbf{x}_m(k\tau)$ is received. These state prediction and measurement are used to update the state estimation $[\lowx(k\tau),\upperx(k\tau)]$. This state estimation and the desired input $\mathbf{u}_{desire}^k$ are the inputs to the supervisor and used to predict the state at the next time step, denoted by $\mathbf{X}_{pred}(\mathbf{u}_{desire}^k(\cdot))$. Given this state prediction, the verification problem is solved. If the answer is \textit{yes}, the supervisor allows the controlled vehicles to travel with the desired input for time $[k\tau,(k+1)\tau)$. A safe input signal $\aggsafe{k+1}$ is generated using a feasible schedule $\mathbf{T}_1$ and stored for a possible use at the next time step. If the answer is \textit{no}, a safe input signal $\aggsafe{k}$ stored at the previous step is used to override the controlled vehicles. It will be proved in this section that the state prediction $\mathbf{X}_{pred}(\aggsafe{k})$ always has a feasible schedule $\mathbf{T}_2$, that is, $\mathbf{X}_{pred}(\aggsafe{k})\in$~Problem~\ref{problem:verification}. A safe input signal $\aggsafe{k+1}$ is generated using $\mathbf{T}_2$ and stored. The output $\aggout{k}$ is either $\mathbf{u}_{desire}^k(\cdot)$ or $\aggsafe{k}$, and the state prediction $\mathbf{X}_{pred}(\aggout{k})$ will be used at the next step to correct the state estimation with a new state measurement.

In the last section, we presented Algorithm~\ref{algorithm:verification} that solves the verification problem. This section describes the remaining components: state prediction, correction of the state estimation, and a safe input generator. Also, an algorithm to implement the supervisor is presented with the proof that this supervisor is the solution of Problem~\ref{problem:supervisor}.

\subsubsection{State prediction.} 

Using the dynamic model \eqref{equation:generalmodel}, this function predicts the state reached at the next time step with a given input signal starting from a set of states. Suppose that at time $k\tau$, we have the state estimation $[\lowx(k\tau), \upperx(k\tau)]$ and the input $\mathbf{u}^k(\cdot)\in\mathcal{U}$ defined on time $[k\tau, (k+1)\tau)$. The state prediction $\mathbf{X}_{pred}(\mathbf{u}^k(\cdot))=[\min\mathbf{X}_{pred}(\mathbf{u}^k(\cdot)), \max \mathbf{X}_{pred}(\mathbf{u}^k(\cdot))]\subset \mathbf{X}$ is a set of possible states at the next time step. Its $i$-th set is denoted by $X_{pred,i}(u^k_i(\cdot))=[\min X_{pred,i}(u^k_i(\cdot)), \max X_{pred,i}(u^k_i(\cdot))]$. The state prediction is defined as follows: for $j\in \mathcal{C}$,
\begin{align}
\begin{split}\label{equation:stateprediction_j}
&\min{X_{pred,j}(u^k_j(\cdot))} = \min_{\signal{d}{j}\in \mathcal{D}_j}
x_j(\tau,u_j^k(\cdot),\signal{d}{j},x^a_j(k\tau)),\\
&\max{X_{pred,j}(u^k_j(\cdot))} = \max_{\signal{d}{j}\in \mathcal{D}_j}
x_j(\tau,u_j^k(\cdot),\signal{d}{j}, x^b_j(k\tau)).
\end{split}
\end{align}
For $\gamma\in\uncontrolled,$
\begin{align}
\begin{split}\label{equation:stateprediction_gamma}
&\min{X_{pred,\gamma}} = \min_{\begin{subarray} 
a\signal{w}{\gamma}\in \mathcal{W}_\gamma,\\
	\signal{d}{\gamma}\in \mathcal{D}_\gamma
\end{subarray}} x_\gamma(\tau,\signal{w}{\gamma},\signal{d}{\gamma},x^a_\gamma(k\tau)),\\
&\max{X_{pred,\gamma}} = \max_{\begin{subarray}
a\signal{w}{\gamma}\in \mathcal{W}_\gamma,\\
\signal{d}{\gamma}\in \mathcal{D}_\gamma
\end{subarray}} x_\gamma(\tau,\signal{w}{\gamma},\signal{d}{\gamma},x^b_\gamma(k\tau)).
\end{split}
\end{align}
By definition, $\min \mathbf{X}_{pred}(\mathbf{u}^k(\cdot))$ is the smallest state propagated from $\lowx(k\tau)$ for time $\tau$ with the input $\mathbf{u}^k(\cdot)$ for any disturbance and any driver-input of uncontrolled vehicles, and $\max \mathbf{X}_{pred}(\mathbf{u}^k(\cdot))$ is the largest.

Notice that \procedureES$(\mathbf{X}_{pred}(\mathbf{u}_{desire}^k(\cdot)),S)$ in Algorithm~\ref{algorithm:verification} determines the existence of an input signal $\aggsafe{k+1,\infty}$ that satisfies $\mathbf{y}(t,\aggsafe{k+1,\infty},\aggsignal{w},\aggsignal{d},\mathbf{x}((k+1)\tau))\notin B$ for all $t$ for any $\aggsignal{w}$ and $\aggsignal{d}$ for any $\mathbf{x}((k+1)\tau)\in \mathbf{X}_{pred}(\mathbf{u}_{desire}^k)$. If we define $\aggdesire{k}$ as
\begin{equation}\label{equation:udesire}
\mathbf{u}_{desire}^{k,\infty}(t) = \begin{cases}
\mathbf{u}_{desire}^k & \text{for}~t\in [k\tau, (k+1)\tau),\\
\mathbf{u}_{safe}^{k+1,\infty}(t) & \text{for}~t\in[(k+1)\tau,\infty),
\end{cases}
\end{equation} then $\mathbf{y}(t,\aggdesire{k},\aggsignal{w},\aggsignal{d},\mathbf{x}(k\tau))\notin B$ for all $t$ for any $\aggsignal{w}$ and $\aggsignal{d}$ for any $\mathbf{x}(k\tau)\in [\lowx(k\tau), \upperx(k\tau)]$. 

\subsubsection{Correction of the state estimation.} 
Once a new state measurement is received, this function restricts the state estimation so that it is compatible with the state measurement and the state prediction computed at the previous step. Suppose the supervisor returns $\aggout{k-1}$ at time $(k-1)\tau$. Then, at time $k\tau$, we have $\mathbf{x}(k\tau)\in \mathbf{X}_{pred}(\aggout{k-1})$ by definitions \eqref{equation:stateprediction_j} and \eqref{equation:stateprediction_gamma}. Also, a new state measurement $\mathbf{x}_m(k\tau)$ is received, which implies $\mathbf{x}(k\tau)\in [\mathbf{x}_m(k\tau)+\delta_{min}, \mathbf{x}_m(k\tau)+\delta_{max}]$. Thus, we make a correction of the state estimation $[\lowx(k\tau), \upperx(k\tau)]$ at time $k\tau$ as the intersection of these two intervals. That is,
\begin{align}\label{equation:statecorrection}
\begin{split}
&\lowx(k\tau) = \max (\min\mathbf{X}_{pred}(\aggout{k-1}),\mathbf{x}_m(k\tau)+\mathbf{\delta}_{min}
),\\
&\upperx(k\tau) =
\min(\max\mathbf{X}_{pred}(\aggout{k-1}),\mathbf{x}_m(k\tau)+{\delta}_{max}).
\end{split}\end{align}
Notice that this correction still guarantees $\mathbf{x}(k\tau)\in[\lowx(k\tau),\upperx(k\tau)]$. The supervisor takes this corrected state estimation as an input as shown in Figure~\ref{figure:supervisor}.


\subsubsection{Safe input generator.} Given a feasible schedule, this function generates a corresponding safe input signal. This is possible because the existence of a feasible schedule implies the existence of a safe input signal by Theorem~\ref{theorem:equivalence}. We define a safe input generator $\sigma(\mathbf{X}_{pred}(\mathbf{u}^k(\cdot)),\mathbf{T})$ to compute $\aggsafe{k+1,\infty}$, where $\mathbf{T}$ is the schedule returned by \procedureES$(\mathbf{X}_{pred}(\mathbf{u}^k(\cdot)),S)$. For $j\in\mathcal{C}$, if $T_j>0$,
\begin{align}\label{equation:safeinput}
\begin{split}
&\sigma_j(X_{pred,j}(u_{j}^k(\cdot)),T_j):=\safe{k+1}{j}\\
&\in\{\signal{u}{j}\in \mathcal{U}_j:\\
&\hspace{0.25 in} y^a_j(T_j+P_j(T_j),\signal{u}{j},
\min X_{pred,j}(u_{j}^k(\cdot)))=\beta_j\\
&\hspace{0.5 in}\text{and}~y^b_j(T_j,\signal{u}{j}, \max X_{pred,j}(u_{j}^k(\cdot)))=\alpha_j\},
\end{split}
\end{align}
where $\sigma_j(X_{pred,j}(u_{j}^k(\cdot)),T_j)$ is the $j$-th entry of $\sigma(\mathbf{X}_{pred}(\mathbf{u}^k(\cdot)),\mathbf{T})$.
If $T_j=0$, then
let ${u}_{safe,j}^{k+1,\infty}(t)=u_{j,max}$ for $t\in [(k+1)\tau,\infty)$. The safe input signal $\safe{k+1}{j}$ makes controlled vehicle $j$ enter the intersection no earlier that $T_j$ and exit it no later than $T_j+P_j(T_j)$. 

If \procedureES$(\mathbf{X}_{pred}(\mathbf{u}_{desire}^k),S)$ finds a feasible schedule $\mathbf{T}$, the supervisor computes a safe input signal $\aggsafe{k+1}$, which is $\mathbf{u}^{k+1,\infty}_{safe}(t)$ restricted to
$t\in [(k+1)\tau, (k+2)\tau)$. The supervisor stores this safe input signal for a possible use at the next time step.

\subsubsection{Solution of Problem~\ref{problem:supervisor}}The supervisor is implemented in procedure~\supervisor~in Algorithm \ref{algorithm:supervisor}.

\begin{algorithm}[H]
	\caption{Implementation of the supervisor}
	\label{algorithm:supervisor}
	\begin{algorithmic}[1]
		\Procedure{Supervisor}{$[\lowx(k\tau),\upperx(k\tau)],\mathbf{u}_{desire}^k$}
		\State $(\mathbf{T}_1,ans )=\procedureES(\mathbf{X}_{pred}(\mathbf{u}_{desire}^k(\cdot)),S)$
		\If{$ans=yes$ and for all $t\in [0,\tau), B\cap[\lowx(t,\mathbf{u}_{desire}^k(\cdot),\lowx(k\tau)),\upperx(t,\mathbf{u}_{desire}^k(\cdot),\upperx(k\tau))]=\emptyset$ }
		\State $\aggsafe{k+1,\infty}= \sigma(\mathbf{X}_{pred}(\mathbf{u}_{desire}^k(\cdot)),\mathbf{T}_1)$
		\State $\aggsafe{k+1}= \mathbf{u}_{safe}^{k+1,\infty}(t)
		~\text{for}~t\in[(k+1)\tau, (k+2)\tau)$
		\State \textbf{return} $\mathbf{u}_{desire}^k(\cdot)$\label{supervisor:return_desire}
		\Else
		\State $(\mathbf{T}_2,\cdot)=\procedureES(\mathbf{X}_{pred}(\aggsafe{k}),S)$\label{supervisor:T2}
		\State $\aggsafe{k+1,\infty}= \sigma(\mathbf{X}_{pred}(\aggsafe{k}),\mathbf{T}_2)$
		\State $\aggsafe{k+1}= \mathbf{u}_{safe}^{k+1,\infty}(t)~
		\text{for}~ t\in[(k+1)\tau, (k+2)\tau)$
		\State \textbf{return} $\aggsafe{k}$\label{supervisor:return_safe}
		\EndIf
		\EndProcedure
	\end{algorithmic}
\end{algorithm}
To initiate the procedure, it is assumed that initial state estimation $[\lowx(0),\upperx(0)]$ and initial desired input $\mathbf{u}_{desire}^0$ do not cause collisions at any future time. That is, \procedureES$(\mathbf{X}_{pred}(\mathbf{u}_{desire}^0),S)$ must return \textit{yes} so that \supervisor$([\lowx(0),\upperx(0)],\mathbf{u}_{desire}^0)\ne \emptyset$.

\begin{theorem}\label{theorem:supervisor}
	Procedure~\supervisor~in Algorithm~\ref{algorithm:supervisor} implements the supervisor $s$ designed in Problem~\ref{problem:supervisor}.
\end{theorem}
\begin{proof}
	Procedure \supervisor~returns $\mathbf{u}_{desire}^k(\cdot)$ in line~\ref{supervisor:return_desire} if $ans_1$ is \textit{yes}. This implies that there exists $\aggsafe{k+1,\infty}$, which in turn, implies by \eqref{equation:udesire} that there exists $\aggdesire{k}$ that satisfies $\mathbf{y}(t,\aggdesire{k},\aggsignal{w},\aggsignal{d},\mathbf{x}(k\tau))\notin B$ for all $t$ for any $\aggsignal{w}$ and $\aggsignal{d}$ for any $\mathbf{x}(k\tau)\in [\lowx(k\tau), \upperx(k\tau)]$. Otherwise, it returns $\aggsafe{k}$ in line~\ref{supervisor:return_safe}. This structure corresponds to the supervisor design in Problem~\ref{problem:supervisor}.
	
	We prove the non-blocking property by mathematical induction on time step $k$. For the base case, it is assumed that \supervisor $([\lowx(0),\upperx(0)],\mathbf{u}_{desire}^0)=\aggout{0}\ne \emptyset$ where $\aggout{0}$ is defined on time $[0,\tau)$, and $\aggsafe{1,\infty}$ is well-defined. We say $\aggsafe{1,\infty}$ is well-defined if there exists a schedule $\mathbf{T}$ that defines $\aggsafe{1,\infty}$ as $\sigma(\mathbf{X}_{pred}(\mathbf{u}_{desire}^0(\cdot)),\mathbf{T})$. Suppose at $t=(k-1)\tau$, we have \supervisor $([\lowx((k-1)\tau),
	\upperx((k-1)\tau)],\mathbf{u}_{desire}^{k-1})=\aggout{k-1}\ne
	\emptyset$ and $\aggsafe{k,\infty}$ is well-defined. That is, there exists $\aggsafe{k,\infty}$ that for all $t\geq 0$ for any $\aggsignal{w}$ and $\aggsignal{d}$,
	\begin{align}
	\begin{split}\label{equation:usafe_at_kt}
	\forall\mathbf{x}(k\tau)\in &\mathbf{X}_{pred}(\aggout{k-1}),\\ &\mathbf{y}(t,\aggsafe{k,\infty},\aggsignal{w}, \aggsignal{d}, \mathbf{x}(k\tau))\notin B.
	\end{split}
	\end{align}
	Then, at $t=k\tau$, we need to show that $\aggout{k}\ne \emptyset$ no matter what $\mathbf{u}_{desire}^k$ is applied, and $\aggsafe{k+1,\infty}$ is well-defined.
	
	In Algorithm~\ref{algorithm:supervisor}, \supervisor$([\lowx(k\tau), \upperx(k\tau)],\mathbf{u}_{desire}^k)$ assigns $\aggout{k}$ either $\mathbf{u}_{desire}^k(\cdot)$ in line~\ref{supervisor:return_desire} or $\aggsafe{k}$ in line~\ref{supervisor:return_safe}. In either case, $\aggout{k}\ne \emptyset$. In the former case, $ans=yes$ and $\mathbf{T}_1$ exists, which implies by Theorem~\ref{theorem:equivalence} that there exists an input signal guaranteeing the avoidance of the Bad set for any uncertainty. Thus, $\aggsafe{k+1,\infty}=\sigma(\mathbf{X}_{pred}(\mathbf{u}_{desire}^k),\mathbf{T}_1)$
	is well-defined on time $[(k+1)\tau, \infty)$. In the latter case, we consider \procedureES$(\mathbf{X}_{pred}(\aggsafe{k}),S)$. Here, $\aggsafe{k}$ is $\aggsafe{k,\infty}$ restricted to time $[k\tau, (k+1)\tau)$. If let $\aggsafe{k+1,\infty}$ be $\aggsafe{k,\infty}$ restricted to time $[(k+1)\tau,\infty)$, we can rewrite \eqref{equation:usafe_at_kt} as 
	\begin{align}
	\forall \mathbf{x}((k+1)\tau) \in [\lowx&(\tau, \aggsafe{k}, \min\mathbf{X}_{pred}(\aggout{k-1})),\notag\\& \upperx(\tau, \aggsafe{k}, \max \mathbf{X}_{pred}(\aggout{k-1}))],\notag\\
	\mathbf{y}(t,\aggsafe{k+1,\infty},&\aggsignal{w}, \aggsignal{d}, \mathbf{x}((k+1)\tau))\notin B.\label{equation:usafe_at_k+1}
	\end{align}
	Since by \eqref{equation:statecorrection}, $[\lowx(k\tau), \upperx(k\tau)]\subseteq \mathbf{X}_{pred}(\aggout{k-1})$, the state prediction $\mathbf{X}_{pred}(\aggsafe{k})$, which denotes $[\lowx(\tau, \aggsafe{k},\lowx(k\tau)), \upperx(\tau, \aggsafe{k},\upperx(k\tau))]$, satisfies
	\begin{align}
	\begin{split}\label{equation:predict_inside}
	\mathbf{X}_{pred}(\aggsafe{k})
	\subseteq [&\lowx(\tau, \aggsafe{k}, \min\mathbf{X}_{pred}(\aggout{k-1})),\\ \upperx & (\tau, \aggsafe{k}, \max \mathbf{X}_{pred}(\aggout{k-1}))],
	\end{split}\end{align} 
	due to the order-preserving property with respect to an initial state in Assumption~\ref{assumption:x_orderpreserving}. Thus, \eqref{equation:usafe_at_k+1} is still satisfied for any $\mathbf{x}((k+1)\tau)\in \mathbf{X}_{pred}(\aggsafe{k})$. That is, $\mathbf{T}_2$ in line~\ref{supervisor:T2} exists, and  $\aggsafe{k+1,\infty}=\sigma(\mathbf{X}_{pred}(\mathbf{u}_{safe}^k),\mathbf{T}_2)$ is well-defined.
	
	Therefore, the supervisor is non-blocking.\qed
\end{proof}

\section{Approximate solutions}\label{section:approxsol}
\subsection{Efficient Verification}\label{section:efficientverification}

While scheduling problems on a single machine with arbitrary release times, deadlines, and process
times are known to be NP-hard, \cite{garey_scheduling_1981} proved that the complexity can be
reduced to $O(n\log n)$ if process times of all jobs are identical. This was done using \textquotedblleft forbidden regions" and the \textquotedblleft earliest deadline scheduling (EDD)" rule. Forbidden regions are time intervals during which no feasible job is allowed to start, and can be computed in $O(n\log n)$ time. Once forbidden regions are computed, EDD can be used to solve the scheduling problem in $O(n\log n)$ time. 


We design Algorithm~\ref{algorithm:polynomial} by modifying Garey's result to handle inserted idle-times. We define initial forbidden regions $\mathbf{F}_0$ to account for the idle-times and set them as inputs of Algorithm~\ref{algorithm:polynomial}. In Garey's result, forbidden regions are initially declared empty and not taken as inputs.

In the procedure in Algorithm~\ref{algorithm:polynomial}, $r_j,d_j$, and $t_j$ denote the $j$-th entry of $\mathbf{r},\mathbf{d}$, and $\mathbf{t}$, respectively, and $F_{0,\gamma}$ the $\gamma$-th interval of $\mathbf{F}_0$ so that $\mathbf{F}_0=\cup_\gamma F_{0,\gamma}$. Critical time $c$ is the latest time at which a job can start at each iteration. A set $\mathcal{A}$ contains jobs that have not been scheduled but are ready to be scheduled at time $s$, which means their release times are smaller than or equal to $s$. A set $\mathcal{B}$ contains all jobs that have not been yet scheduled. Initially, $\mathcal{B}=\{1,\ldots,|\mathbf{r}|\}$ where $|\mathbf{r}|$ is the cardinality of $\mathbf{r}$. 

\begin{algorithm}[H]
\begin{algorithmic}[1]
\Procedure{\garey}{$\mathbf{r,d,F_0}$}
\State Let $\sigma$ be a vector of indexes in an increasing order of $\mathbf{r}$ and $\mathbf{F}=\mathbf{F}_0$.
\For{$i=|\mathbf{r}|$ to 1} \label{polynomial:forbidden_start}\Comment{{Forbidden Region Declaration}}
	\For{$j\in \{j:d_j\geq d_{\sigma_i}\}$} 
	\If{$c_j$ undefined} {$c_j = d_j$} \Else~{$c_j=c_j-1$} \EndIf \label{polynomial:declaration_minus1}
	\If{$c_j\in F_\gamma$ for some $\gamma$} {$c_j=\inf{F
			_\gamma}$}\EndIf\label{polynomial:declaration_c}
	\EndFor
	\If{$\sigma_i=1$ or $r_{\sigma_{i-1}}<r_{\sigma_i}$} 
	\State {$c=\min \{c_j:c_j~\text{defined}\}$} \EndIf
	\If{$c<r_{\sigma_i}$} {$ans = no$}\EndIf \label{polynomial:ans=no}
	\If{$r_{\sigma_i}\leq c<r_{\sigma_i}+1$} {$\mathbf{F}=\mathbf{F}\cup (c-1,r_{\sigma_i})$}
	\EndIf
\EndFor\label{polynomial:forbidden_end}
\State $s= 0, \mathcal{A}=\emptyset, \mathcal{B}=\{1,\ldots, |\mathbf{r}|\}$ \label{polynomial:initialize} \Comment{{Schedule Generation}}
\While{$\mathcal{B} \ne \emptyset$}\label{polynomial:while}
	\If{$s\in {F}_\gamma$ for some $\gamma$} {$s=\sup {F}_\gamma$}
	\EndIf\label{polynomial:forbidden}
	\If{$\mathcal{A}=\emptyset$} {$j= \argmin_{j\in \mathcal{B}} r_j$ and $t_{j}= r_{j}$}
\label{polynomial:if_A_empty_T=r}
	\Else ~{$j= \argmin_{j\in \mathcal{A}} d_j$ and $t_{j}=
s$}\label{polynomial:if_A_non_empty_T=s}
	\EndIf
	\State $s= s+1,~\mathcal{B}=\mathcal{B}\setminus\{j\}$ and $\mathcal{A}= \{k\in \mathcal{B}:r_k\geq
s\}$\label{polynomial:ABupdate}
\EndWhile\label{polynomial:end_while}
\State{$\pi^*=$ a vector of indexes in an increasing order of $\mathbf{t}$}
\If{$ans=no$} \label{polynomial:feasibility}
\State \textbf{return} $(\emptyset,\pi^*,no)$
\Else~{\textbf{return} $(\mathbf{t},\pi^*,yes)$}
\EndIf
\EndProcedure\label{polynomial:lastline}
\end{algorithmic}
\caption{Modified version of the result of \cite{garey_scheduling_1981}}
\label{algorithm:polynomial}
\end{algorithm}

In the following lemma, we prove that procedure \garey~in Algorithm~\ref{algorithm:polynomial} solves the IIT scheduling problem with unit process times. This problem is formulated as follows.

\begin{problem}\label{problem:unit_process_scheduling}
Given $\mathbf{r,d,\bar{r},\bar{p}}$, determine the existence of a schedule $\mathbf{t}$ satisfying 
\begin{align}
&\text{for all}~j, && t_j\in [r_j, d_j],\label{Unit-IIT:bounded} \\
&\text{for all}~i\ne j, &&(t_i,t_i+1)\cap (t_j,t_j+1)=\emptyset,\label{Unit-IIT:disjoint}\\
&\text{for all}~j~\text{and}~\gamma, && (t_j,t_j+1)\cap (\bar{r}_\gamma, \bar{p}_\gamma)=\emptyset, \label{Unit-IIT:iit}
\end{align}
where $(\bar{r}_\gamma, \bar{p}_\gamma)$ denotes the $\gamma$-th inserted idle-time. 
\end{problem}

\begin{lemma}\label{lemma:garey_solve_iit}
Procedure \garey~in Algorithm~\ref{algorithm:polynomial} solves the IIT scheduling problem with unit process times and finds a feasible schedule if exists.
\end{lemma}

The proof of Lemma~\ref{lemma:garey_solve_iit} is provided in Appendix B.

To use Algorithm~\ref{algorithm:polynomial}, we assign a time interval of equal length to cross the intersection to all vehicles, and formulate a \textit{relaxed} IIT scheduling problem. The identical process time $\theta_{max}$ is defined as follows:
\begin{equation}
\theta_{max} := \max_{j\in \mathcal{C}}\max_{R_j\leq T_j\leq D_j} P_j(T_j).\label{equation:themamax}
\end{equation}
Here, $\theta_{max}$ is the maximum time that any controlled vehicle
spends crossing an intersection, so that $\theta_{max}\geq P_j(T_j)$ for all
$j\in\mathcal{C}$ for any $T_j\in[R_j,D_j]$. In other words, all controlled vehicles are guaranteed to cross the intersection within $\theta_{max}$.

By replacing
$P_j(T_j)$ in Problem~\ref{problem:scheduling} with $\theta_{max}$, the relaxed IIT scheduling problem is formulated as follows.

\begin{problem}[(Relaxed IIT scheduling)]\label{problem:relaxedS}
	Given an initial state estimation $[\lowx(0), \upperx(0)]$, determine whether
there exists a schedule $\mathbf{T}=(T_1,\ldots, T_{n_c}) \in
\mathbb{R}^{n_c}_{+}$ such that for all $j\in \mathcal{C}$, 
	\begin{equation}
	\label{condition:efficient_boundedinput}
	R_j\leq T_j\leq D_j,
	\end{equation}
	for all $i\ne j \in \mathcal{C}$ if $T_i, T_j > 0$,
	\begin{equation}
	(T_i,T_i +\theta_{max})\cap (T_j,
T_j+\theta_{max})=\emptyset,\label{condition:efficient_controlled}
	\end{equation}
	for all $j\in \mathcal{C}$ and $\gamma\in\uncontrolled$ if $T_j > 0$,
	\begin{equation}
	(T_j, T_j+\theta_{max})\cap (\bar{R}_\gamma
,\bar{P}_{\gamma})=\emptyset.\label{condition:efficient_uncontrolled}
	\end{equation}
	If $T_j = 0$, $(0,P_j(0))$ replaces $(T_j,T_j+\theta_{max})$ in conditions~\eqref{condition:efficient_controlled} and \eqref{condition:efficient_uncontrolled}.
\end{problem}

The following algorithm solves this problem by employing \garey. This is an exact solution for Problem~\ref{problem:relaxedS} by Lemma~\ref{lemma:garey_solve_iit}. At line~\ref{relaxed:zero_vec}, we call $\mathbf{0}$ the zero vector in $\mathbb{R}^{n_c}$. 

\begin{algorithm}[H]
\begin{algorithmic}[1]
\Procedure{\textls[-80]\procedureRES}{$[\lowx(0), \upperx(0)],S$}
\If{$[\mathbf{y}^a(0), \mathbf{y}^b(0)]\cap B\ne \emptyset$}~\textbf{return} $(\emptyset,\emptyset,no)$\label{relaxed:inside_B}
\EndIf
\ForAll{$j\in\mathcal{C}$} {calculate $R_j,D_j, \theta_{max}$}
\EndFor
\ForAll{$\gamma\in\uncontrolled$} {calculate
$\bar{R}_{\gamma},\bar{P}_{\gamma}$}
\State $F_\gamma =  (\max(\bar{R}_{\gamma}/\theta_{max}-1,0),\bar{P}_{\gamma}/\theta_{max})$\label{relaxed:forbidden}
\EndFor
\State $\mathcal{M} = \{j\in\mathcal{C}: y^b_j(0)<\alpha_j\}$ and $\bar{\mathcal{M}}=\mathcal{C}\setminus \mathcal{M}$ 
\If{$|\mathcal{M}| = 0$} \textbf{return} $(\emptyset,\mathbf{0},yes)$ \EndIf\label{relaxed:zero_vec}
\State $T_j = 0$ for all $j\in\bar{\mathcal{M}}$ and $P_{max} = \max_{j\in \bar{\mathcal{M}}} P_j(T_j)$
\ForAll{$j\in{\mathcal{M}}$} 
\State $r_j = \max(R_j,P_{max})/\theta_{max}$ and $d_j = D_j/\theta_{max}$
\EndFor
\State $(\mathbf{t},\pi^*,ans) = \textsc{Polynomial}(\mathbf{r,d,F})$\label{relaxed:garey}
\If{$ans = yes$}  $T_j = t_j\theta_{max}, \forall j\in{\mathcal{M}}$
\State \textbf{return} $(\mathbf{T},\pi^*,yes)$
\Else~ \textbf{return} $(\emptyset,\pi^*,no)$
\EndIf

\EndProcedure\label{effscheduling:lastline}
\end{algorithmic}
\caption{Exact Solution of Problem~\ref{problem:relaxedS}}
\label{algorithm:relaxedexact}
\end{algorithm}

In this procedure, all parameters are normalized by $\theta_{max}$ because procedure \garey~assumes unit process times. In line~\ref{relaxed:forbidden}, the idle-time $(\bar{R}_\gamma, \bar{P}_\gamma)$ is translated into an initial set of forbidden regions $F_\gamma$ so that condition~\eqref{condition:efficient_uncontrolled} is equivalent to $t_j\notin F_{\gamma}$. If $t_j\in F_\gamma$, then $T_j\in (\max(\bar{R}_\gamma - \theta_{max},0), \bar{P}_\gamma)$ so that either $T_j\in (\bar{R}_\gamma, \bar{P}_\gamma)$ or $T_j+\theta_{max}\in (\bar{R}_\gamma, \bar{P}_\gamma)$.

 The running time of procedure \procedureRES~is dominated by \garey~in line~\ref{relaxed:garey}, which has an
asymptotic running time of $O(n_c^2)$. 

By exploiting procedure \procedureRES~in Algorithm \ref{algorithm:relaxedexact}, we design a new procedure, called \procedureAS, that solves Problem~\ref{problem:verification} within an approximation bound. This procedure schedules vehicles according to a sequence returned by procedure \procedureRES, thereby inheriting computational efficiency. This sequence is denoted by $\pi^*$ in the following algorithm. 

\begin{algorithm}[H]
\begin{algorithmic}[1]	
\Procedure{\procedureAS}{$[\lowx(0), \upperx(0)],S$}
\If{$[\mathbf{y}^a(0), \mathbf{y}^b(0)]\cap B\ne \emptyset$} \textbf{return} $(\emptyset,no)$\label{approx:inside_B}
\EndIf
\If{$y^b_j(0)\geq\alpha_j$ for all $j\in\mathcal{C}$} \textbf{return} $(\emptyset,yes)$\label{approx:all_crossed}
\EndIf
\State $(\cdot,\pi^*, \cdot) = \procedureRES([\lowx(0),\upperx(0)],S)$\label{approx:RelaxedExact}
\State $(\mathbf{T},ans)=\procedureS(\pi^*,[\lowx(0),\upperx(0)],S)$\label{approx:scheduling}
\EndProcedure
\end{algorithmic}
\caption{Approximate Solution of Problem~\ref{problem:scheduling}}
\label{algorithm:approximate}
\end{algorithm}

Procedure \procedureAS~in Algorithm~\ref{algorithm:approximate} trades exactness for computational
efficiency. We will prove that procedure \procedureAS~is more conservative than procedure \procedureES, that is, there exists an instance $I=([\lowx(0), \upperx(0)],S)$ such that \procedureAS$(I)$ returns \textit{no}~while \procedureES$(I)$ returns \textit{yes}. In order to quantify the degree of conservatism, we prove two theorems. The
first theorem states that if procedure \procedureAS~returns
\textit{yes}, then there exists an input signal to avoid the Bad set. In the second theorem, if procedure \procedureAS~returns \textit{no}, then there does not exist an input signal to avoid an inflated Bad set, which accounts for the conservatism.

 \begin{theorem}\label{theorem:approxS_yes_V_yes}
 If \procedureAS$([\lowx(0),\upperx(0)],S)=(\mathbf{T},yes)$, then
$([\lowx(0),\upperx(0)],S)\in$ Problem~\ref{problem:verification}. 
 \end{theorem}
 \begin{proof}
 By Theorem~\ref{theorem:equivalence}, we only need to show that
$([\lowx(0),\upperx(0)],S)\in$ Problem~\ref{problem:scheduling}. Notice that procedure \procedureAS~of Algorithm~\ref{algorithm:approximate} returns \textit{yes} in the case of line~\ref{approx:all_crossed}. In this case, $([\lowx(0),\upperx(0)],S)\in$ Problem~\ref{problem:scheduling} since all vehicles have already crossed the intersection. The procedure also returns \textit{yes} when procedure \procedureS~given a sequence $\pi^*$ returns \textit{yes} in line~\ref{approx:scheduling}.
Since $\pi^*\in \mathcal{P}$, where $\mathcal{P}$ is a set of all permutations, and $\pi^*$ yields a feasible schedule, we know \procedureES$([\lowx(0),\upperx(0)],S)$ returns \textit{yes}. Thus, $([\lowx(0),\upperx(0)],S)\in$ Problem~\ref{problem:scheduling}.
 \end{proof}

 However, the converse of Theorem~\ref{theorem:approxS_yes_V_yes} is not true. That is, for some instances
such that $([\lowx(0),\upperx(0)],S)\in$ Problem~\ref{problem:verification}, procedure \procedureAS~can return \textit{no}. To consider these
instances, we introduce an \textit{inflated} Bad set.

The IIT scheduling problem finds a schedule satisfying $y_j^b(T_j)=\alpha_j$ and $y_j^a(T_j+P_j(T_j))=\beta_j$ for $j\in\mathcal{C}$. In contrast, the relaxed IIT scheduling problem finds a schedule satisfying $y_j^b(T_j)=\alpha_j$ and $y_j^a(T_j+\theta_{max})\geq {\beta}_j$. Given that the farthest distance that controlled vehicle $j$ can travel during $\theta_{max}$ is $\theta_{max}v_{j,max}$, we define an
\textquotedblleft inflated" intersection $(\alpha_j, \hat{\beta}_j)$ such that $\hat{\beta}_j:=\alpha_j
+ \theta_{max}v_{j,max}$. Notice that $\beta_j\leq \hat{\beta}_j$. Because the process times are only defined for controlled vehicles, $\hat{\beta}_\gamma = \beta_\gamma$ for $\gamma\in\uncontrolled$. Thus, inflated Bad
set $\hat{B}$ is defined as follows:
\begin{align}\label{equation:inflatedB}
\begin{split}
\hat{B}:=\{&\mathbf{y}\in Y:
y_i\in(\alpha_i,\hat{\beta}_i)~\text{and}~y_j\in(\alpha_j,\hat{\beta}_j)\\
&\text{for some}~i\ne j~\text{such that}~i\in \mathcal{\allset}~\text{and}~j\in
\mathcal{C}\}.
\end{split}
\end{align}
By replacing Bad set $B$ in Problem~\ref{problem:verification} with inflated Bad set $\hat{B}$, we can formulate the relaxed verification problem. The following theorem is a key result that shows the approximation bound of procedure~\procedureAS.

\begin{lemma}\label{lemma:T<barT}
If \procedureAS$([\lowx(0),\upperx(0)],S)=(\mathbf{T},yes)$, and 
\procedureRES$([\lowx(0),\upperx(0)],S)=(\bar{\mathbf{T}},\pi^*,yes)$, then
$T_{j}\leq \bar{T}_{j}$ for all $j\in\mathcal{C}$.
\end{lemma}

\begin{lemma}\label{lemma:Appro_no_Relaxed_no}
	If \procedureAS$([\lowx(0),\upperx(0)],S)=(\emptyset,no)$, then \procedureRES$([\lowx(0),\upperx(0)],S)=(\emptyset,\pi^*,no)$.
\end{lemma}

The proofs of Lemmas~\ref{lemma:T<barT} and \ref{lemma:Appro_no_Relaxed_no} can be found in Appendix B. 

\begin{theorem}\label{theorem:approxS_no_relaxedV_no}
	If \procedureAS$([\lowx(0),\upperx(0)],S)=(\emptyset,no)$, then there is no input signal $\aggsignal{u}$ that guarantees $\mathbf{y}(t,\aggsignal{u},\aggsignal{w},\aggsignal{d},\mathbf{x}(0))\notin \hat{B}$ for all $t\geq 0$ for any $\aggsignal{w}\in\mathcal{W},\aggsignal{d}\in\mathcal{D},$ and $\mathbf{x}(0)\in[\lowx(0),\upperx(0)]$.
\end{theorem}

\begin{algorithm*}
	\caption{Implementation of $s_e$}
	\label{algorithm:effsupervisor}
	\begin{algorithmic}[1]
		\Procedure{\efficientS}{$[\lowx(k\tau),\upperx(k\tau)],\mathbf{u}_{desire}^k$}
		\State $(\mathbf{T}_1,ans_1 )=\procedureAS(\mathbf{X}_{pred}(\mathbf{u}_{desire}^k(\cdot)),S)$\label{effsupervisor:T1_ans1}
		\If{$ans_1=yes$ and for all $t\in [0,\tau), B\cap[\lowx(t,\mathbf{u}_{desire}^k(\cdot),\lowx(k\tau)),\upperx(t,\mathbf{u}_{desire}^k(\cdot),\upperx(k\tau))]=\emptyset$ }
		\State $\pi^k =$ a vector of indexes in the increasing order of nonzero entries of $\mathbf{T}_1$\label{effsupervisor:sequence_save1}
		\State $\aggsafe{k+1,\infty}= \sigma(\mathbf{X}_{pred}(\mathbf{u}_{desire}^k(\cdot)),\mathbf{T}_1)$\label{effsupervisor:yes_usafe}
		\State $\aggsafe{k+1}= \mathbf{u}_{safe}^{k+1,\infty}(t)
		~\text{for}~ t\in[(k+1)\tau, (k+2)\tau)$
		\State \textbf{return} $\mathbf{u}_{desire}^k(\cdot)$\label{effsupervisor:return_desire}
		\Else
		\State $(\mathbf{T}_2,ans_2)=\procedureAS(\mathbf{X}_{pred}(\aggsafe{k}),S)$\label{effsupervisor:T2_ans2}
		\If{$ans_2 = no$}
		\State $(\mathbf{T}_2,ans_3)=\procedureS(\pi^{k-1},\mathbf{X}_{pred}(\aggsafe{k}),S)$\label{effsupervisor:T2_ans3}
		\EndIf
		\State $\pi^k=$ a vector of indexes in the increasing order of nonzero entries of $\mathbf{T}_2$ \label{effsupervisor:sequence_save2}
		\State $\aggsafe{k+1,\infty}= \sigma(\mathbf{X}_{pred}(\aggsafe{k}),\mathbf{T}_2)$\label{effsupervisor:no_usafe}
		\State $\aggsafe{k+1}= \mathbf{u}_{safe}^{k+1,\infty}(t)~
		\text{for}~ t\in[(k+1)\tau, (k+2)\tau)$
		\State \textbf{return} $\aggsafe{k}$\label{effsupervisor:return_safe}
		\EndIf
		\EndProcedure
	\end{algorithmic}
\end{algorithm*}

\begin{proof}
	By Lemma~\ref{lemma:Appro_no_Relaxed_no}, \procedureAS$([\lowx(0),\upperx(0)],S)=(\emptyset,no)$ means that \procedureRES~returns \textit{no}. Thus, we will prove that if there is no schedule satisfying conditions~\eqref{condition:efficient_boundedinput}-\eqref{condition:efficient_uncontrolled}, there is no input signal to avoid the inflated Bad set for any uncertainty. For ease of proof, we prove the contrapositive statement. That is, assume that there is an input signal $\aggsignal{\tilde{u}}\in\mathcal{U}$ such that $\mathbf{y}(t,\aggsignal{\tilde{u}},\aggsignal{w},
	\aggsignal{d},\mathbf{x}(0))\notin \hat{B}$ for all $t$ for any
	$\aggsignal{w}\in\mathcal{W},~\aggsignal{d}\in \mathcal{D}$, and $\mathbf{x}(0)\in[\lowx(0),\upperx(0)]$. Then, there exists schedule $\tilde{\mathbf{T}}$ satisfying conditions~\eqref{condition:efficient_boundedinput}-\eqref{condition:efficient_uncontrolled}. This proof is similar to the first part of the proof of Theorem~\ref{theorem:equivalence}.
	
	For all $j\in\mathcal{C}$, define $\tilde{T}_j$ as $y^b_j(\tilde{T}_j,\signal{\tilde{u}}{j},x^b_j(0))=\alpha_j$ if $y^b_j(0)< \alpha_j$, and $0$ otherwise, and $\tilde{K}_j$ as $y^a_j(\tilde{K}_j,\signal{\tilde{u}}{j},x^a_j(0))=\hat{\beta}_j$ if $y^a_j(0)<\hat{\beta}_j$, and $0$ otherwise. For $i\ne j\in\mathcal{C}$, since $y_i(t,\signal{\tilde{u}}{i},\signal{d}{i},x_i(0))$ and $y_j(t,\signal{\tilde{u}}{j},\signal{d}{j},x_j(0))$ avoid the inflated Bad set for any $\aggsignal{d}\in\mathcal{D}$ and $\mathbf{x}(0)\in[\lowx(0), \upperx(0)]$, $(\tilde{T}_i,\tilde{K}_i)\cap(\tilde{T}_j,\tilde{K}_j)=\emptyset$. Since the inflated intersection takes account of the maximum driving distance during $\theta_{max}$, we have $y^a_j(\tilde{T}_j+\theta_{max},\signal{\tilde{u}}{j},x^a(0))\leq \hat{\beta}_j$. This implies $\tilde{T}_j+\theta_{max}\leq \tilde{K}_j$ for all $j\in\mathcal{C}$ because $y_j^a(t)$ is non-decreasing in $t$. Thus, $(\tilde{T}_i, \tilde{T}_i+\theta_{max})\cap (\tilde{T}_j,\tilde{T}_j+\theta_{max})=\emptyset$ (condition~\eqref{condition:efficient_controlled}).
	
	For $\gamma\in\uncontrolled$, define $(\bar{R}_\gamma, \bar{P}_\gamma)$ as in Definition~\ref{definition:schedulingparmeters}. Then, for any $\signal{w}{\gamma}\in\mathcal{W}_\gamma,\signal{d}{\gamma}\in\mathcal{D}_\gamma,$ and $x_\gamma(0)\in [x_\gamma^a(0),x_\gamma^b(0)]$, uncontrolled vehicle $\gamma$ enters the intersection no earlier than $\bar{R}_\gamma$ and exits it no later than $\bar{P}_\gamma$. In order for $\signal{\tilde{u}}{j}$ to guarantee that $y_j(t, \signal{\tilde{u}}{j},\signal{d}{j},x_j(0))$ and $y_\gamma(t)$ never meet inside the intersection for any uncertainty,
	$(\tilde{T}_j, \tilde{K}_j)\cap(\bar{R}_\gamma, \bar{P}_\gamma)=\emptyset$. Since $\tilde{T}_j+\theta_{max}\leq \tilde{K}_j$, we have $(\tilde{T}_j, \tilde{T}_j+\theta_{max})\cap (\bar{R}_\gamma, \bar{P}_\gamma)=\emptyset$ (condition~\eqref{condition:efficient_uncontrolled}).
	
	 Condition~\eqref{condition:efficient_boundedinput} is satisfied by the definitions of $R_j$ and $D_j$. \qed
\end{proof}

In summary, Theorem~\ref{theorem:approxS_yes_V_yes} states that if procedure \procedureAS~returns \textit{yes}, there exists an input signal to avoid the Bad set $B$ for any uncertainty. As stated in Theorem~\ref{theorem:approxS_no_relaxedV_no}, if the procedure returns \textit{no}, there does not exist an input signal to avoid the inflated Bad set $\hat{B}$ for any uncertainty. Thus, the difference between $B$ and $\hat{B}$ is a measure of the approximation bound of procedure \procedureAS.

\subsection{Efficient supervisor}\label{section:effsupervisor}
In order for a supervisor to run in real time, the verification problem must be solved within the time step $\tau$. However, if a large number of controlled vehicles is considered, the problem becomes intractable. Thus, we employ the results in Section~\ref{section:efficientverification} to design efficient supervisors. The simplest solution would be to implement a supervisor $\hat{s}([\lowx(k\tau),\upperx(k\tau)],\mathbf{u}_{desire}^k)$ defined as follows:
$$
\hat{s} = 
\begin{cases}
\mathbf{u}_{desire}^k(\cdot) & \text{if}~\exists \aggdesire{k}:\text{
	for all}~ t\geq 0\\ 
&\mathbf{y}(t,\aggdesire{k},\aggsignal{w},\aggsignal{d},\mathbf{x}(k\tau))\notin \hat{B}\\
\aggsafe{k} & \text{otherwise.}
\end{cases}
$$
The only difference from the exact supervisor $s$ is that the Bad set $B$ is replaced by the inflated Bad set $\hat{B}$.

Instead of implementing $\hat{s}$, we consider a procedure implementing another supervisor, $s_e$, by using procedure \procedureAS~to verify whether or not to override drivers. We call this procedure \efficientS~in Algorithm~\ref{algorithm:effsupervisor}.

 Notice that this procedure stores the feasible sequence of vehicles crossing the intersection at every time step (lines~\ref{effsupervisor:sequence_save1} and \ref{effsupervisor:sequence_save2}). This is because when the sequence considered in line~\ref{effsupervisor:T2_ans2} does not yield a feasible schedule, the previous step's sequence $\pi^{k-1}$ can be used to generate one. This is where in procedure \procedureS~in Algorithm~\ref{algorithm:verification}, $|\pi_0|$ can be different from $|\pi|$, where $\pi_0=\pi^{k-1}$ and $\pi\in\mathbb{R}^{|\mathcal{M}|}$. This is because  $\mathcal{M}$ is a set of vehicles that have not entered an intersection at the current step, whereas $\pi^{k-1}$ is a sequence of vehicles that had not entered an intersection at the previous step.

The fact that the previous step's sequence leads to a feasible schedule ensures the non-blocking property of the supervisor and is proved in the next theorem. \cite{bruni_robust_2013} proposed an efficient supervisor considering measurement errors and unmodeled dynamics with all vehicles controlled. Their supervisor cannot find a feasible schedule at every step and thus uses the previous schedule until a new feasible schedule is found. This ignores the correction of the state estimation during the open-loop control, thereby being more conservative than our efficient supervisor, which updates the schedule based on the most current state estimation.

Procedure \efficientS~takes polynomial computation time, and guarantees avoiding the Bad set $B$ because such a safe input exists if procedure \procedureAS~returns \textit{yes} by Theorem~\ref{theorem:approxS_yes_V_yes}. Most importantly, we will prove in the following theorem that $s_e$ is less restrictive than $\hat{s}$, in the sense that $\hat{s}$ overrides controlled vehicles more frequently than $s_e$.

\begin{theorem}\label{theorem:efficientsupervisor}
The supervisor $s_e$ is less restrictive than $\hat{s}$: if \efficientS$([\lowx(k\tau),\upperx(k\tau)],\mathbf{u}_{desire}^k)=\aggsafe{k}$, then $\hat{s}([\lowx(k\tau),\upperx(k\tau)],\mathbf{u}_{desire}^k)=\aggsafe{k}$. Moreover, $s_e$ is non-blocking.
\end{theorem}

\begin{proof}
\efficientS $([\lowx(k\tau),\upperx(k\tau)],\mathbf{u}_{desire}^k)=\aggsafe{k}$ when \procedureAS$(\mathbf{X}_{pred}(\mathbf{u}_{desire}^k(\cdot)),S)$ returns \textit{no}. By Theorem~\ref{theorem:approxS_no_relaxedV_no}, there is no input signal $\aggdesire{k}\in \mathcal{U}$ that can prevent entering the inflated Bad set for some uncertainties. Thus, $\hat{s}(\mathbf{X}_{pred}(\mathbf{u}_{desire}^k),S)$ returns $\aggsafe{k}$. This proves that $s_e$ is less restrictive than $\hat{s}$.

The non-blocking property is proved by mathematical induction on time step $k$. For the base case, it is assumed that \efficientS $([\lowx(0),\upperx(0)],\mathbf{u}_{desire}^0)=\aggout{0}\ne \emptyset$ where $\aggout{0}$ is defined on time $[0,\tau)$, and $\aggsafe{1,\infty}$ is well-defined, that is, there exists a schedule $\mathbf{T}$ that defines $\aggsafe{1,\infty}$ as $\sigma(\mathbf{X}_{pred}(\mathbf{u}_{desire}^0(\cdot)),\mathbf{T})$.
Suppose at $t=(k-1)\tau$, we have \efficientS $([\lowx((k-1)\tau),
\upperx((k-1)\tau)],\mathbf{u}_{desire}^{k-1})=\aggout{k-1}\ne
\emptyset$ and $\aggsafe{k,\infty}$ is well-defined, which satisfies $\mathbf{y}(t,\aggsafe{k,\infty},\aggsignal{w}, \aggsignal{d}, \mathbf{x}(k\tau))\notin B$ for any $\aggsignal{w}\in\mathcal{W}, \aggsignal{d}\in\mathcal{D}$, and $\mathbf{x}(k\tau)\in \mathbf{X}_{pred}(\aggout{k-1})$. 
Now, we need to prove that $\aggout{k}\ne \emptyset$, and $\aggsafe{k+1,\infty}$ is well-defined. 

In Algorithm~\ref{algorithm:effsupervisor}, if $ans_1=yes$ in line~\ref{effsupervisor:T1_ans1} or $ans_2=yes$ in line~\ref{effsupervisor:T2_ans2}, by Theorem~\ref{theorem:approxS_yes_V_yes}, there exists an input signal that makes the vehicle trajectories avoid entering the Bad set. Thus, $\aggsafe{k+1,\infty}$
is well-defined on time $[(k+1)\tau, \infty)$ in lines~\ref{effsupervisor:yes_usafe} and \ref{effsupervisor:no_usafe}. Also, $\aggout{k}= \emptyset$ because $\aggout{k}=\aggdesire{k}$ or $\aggout{k}=\aggsafe{k}$. 

Suppose that $ans_1=no$ and $ans_2=no$. Then, given $\pi^{k-1}$, we need to prove that $ans_3=yes$ and $\mathbf{T}_2$ exists in line~\ref{effsupervisor:T2_ans3}. Notice that $\pi^{k-1}$ is a vector of indexed in the increasing order of the nonzero entries of a feasible schedule $\mathbf{T}^{k-1}$ of the previous step. That is, at the previous step, \procedureS$(\pi^{k-1},\mathbf{X}_{pred}(\aggout{k-1}),S)=(\mathbf{T}^{k-1},yes)$. We will show that the existence of $\mathbf{T}^{k-1}$ implies that of $\mathbf{T}_2$.

Let $\aggsafe{k+1,\infty}$ be $\aggsafe{k,\infty}$ restricted to time $[(k+1)\tau,\infty)$ and $\aggsafe{k}$ be $\aggsafe{k,\infty}$ restricted to time $[k\tau, (k+1)\tau)$. Then, the $j$-th entry of $\mathbf{T}^{k-1}$ is as follows:\begin{align}
&T_j^{k-1}:=\{t:y^b_j(t,\safe{k}{j},\max X_{pred,j}(u_{out,j}^{k-1}(\cdot)))=\alpha_j\}\notag\\
&=\{t:y^b_j(t,\safe{k+1}{j},\notag\\
&\hspace{0.5 in}x^b(\tau,u_{safe,j}^k(\cdot),\max X_{pred,j}(u_{out,j}^{k-1}(\cdot)))=\alpha_j\}+\tau\notag\\
&\leq \{t:y^b_j(t,\safe{k+1}{j},x^b(\tau,u_{safe,j}^k(\cdot),x^b(k\tau)))=\alpha_j\}+\tau\notag\\
&:=\tilde{T}_j +\tau.\label{effsupervisor_nonblock:T}
\end{align}
The inequality is due to $[x^a_j(k\tau), x^b_j(k\tau)]\subseteq {X}_{pred,j}(u_{out,j}^{k-1}(\cdot))$ by \eqref{equation:statecorrection} and the order-preserving property in Assumption~\ref{assumption:x_orderpreserving}. Similarly, its process time $T_j^{k-1}+P_j^{k-1}(T_j^{k-1})$ is as follows:
\begin{align}
&T_j^{k-1}+P_j^{k-1}(T_j^{k-1})\notag\\
&:=\{t:y^a_j(t,\safe{k}{j},\min X_{pred,j}(u_{out,j}^{k-1}(\cdot)))=\beta_j\}\notag\\
&=\{t:y^a_j(t,\safe{k+1}{j},\notag\\
&\hspace{0.5 in}x^a(\tau,u_{safe,j}^k(\cdot),\min X_{pred,j}(u_{out,j}^{k-1}(\cdot)))=\beta_j\}+\tau\notag\\
&\geq \{t:y^a_j(t,\safe{k+1}{j},x^a(\tau,u_{safe,j}^k(\cdot),x^a(k\tau)))=\beta_j\}+\tau\notag\\
&:=\tilde{T}_j+\tilde{P}_j(\tilde{T}_j) +\tau.\label{effsupervisor_nonblock:T+P}
\end{align}
For $\gamma\in\uncontrolled$, let $(\bar{R}_\gamma^{k-1}, \bar{P}_\gamma^{k-1})$ denote the idle-time given $\mathbf{X}_{pred}(\aggout{k-1})$.

We now show that the schedule $\tilde{\mathbf{T}}:=(\tilde{T}_j:j\in\mathcal{C})$ is a feasible schedule of \procedureS$(\pi^{k-1}, \mathbf{X}_{pred}(\aggsafe{k}),S)$. The release time and deadline of this procedure for $j\in\mathcal{C}$ is by definition, \begin{align*}
&R_j=\min_{\signal{u}{j}\in\mathcal{U}_j} \{t:y^b_j(t,\signal{u}{j},\max X_{pred,j}(u_{safe,j}^{k}(\cdot)))=\alpha_j\},\\
&D_j=\max_{\signal{u}{j}\in\mathcal{U}_j} \{t:y^b_j(t,\signal{u}{j},\max X_{pred,j}(u_{safe,j}^{k}(\cdot)))=\alpha_j\}.
\end{align*} Since $\upperx(\tau, \aggsafe{k},\upperx(k\tau))=\max \mathbf{X}_{pred}(\aggsafe{k})$ in \eqref{equation:stateprediction_j}, we have $\tilde{T}_j \in [R_j, D_j]$.

Since $\mathbf{T}^{k-1}$ is feasible, $(T_j^{k-1}, T_j^{k-1}+P_j^{k-1}(T_j^{k-1}))\cap (T_i^{k-1},T_i^{k-1}+P_i^{k-1}(T_i^{k-1}))=\emptyset$ for $i\ne j\in\mathcal{C}$. Because of \eqref{effsupervisor_nonblock:T} and \eqref{effsupervisor_nonblock:T+P}, 
$$(\tilde{T}_j, \tilde{T}_j+\tilde{P}_j(\tilde{T}_j))\subseteq(T_j^{k-1}, T_j^{k-1}+P_j^{k-1}(T_j^{k-1}))-\tau,$$
for all $j\in\mathcal{C}$. Thus, $(\tilde{T}_j, \tilde{T}_j+\tilde{P}_j(\tilde{T}_j))\cap (\tilde{T}_j, \tilde{T}_i+\tilde{P}_i(\tilde{T}_i))=\emptyset$.

Similarly for $\gamma\in\uncontrolled$, it is not difficult to see that the idle-time of \procedureS$(\pi^{k-1}, \mathbf{X}_{pred}(\aggsafe{k}),S)$ denoted by $(\bar{R}_\gamma, \bar{P}_\gamma)$ becomes a subset of $(\bar{R}_\gamma^{k-1}, \bar{P}_\gamma^{k-1})-\tau$. Since $(T_j^{k-1}, T_j^{k-1}+P_j^{k-1}(T_j^{k-1}))\cap (\bar{R}_\gamma^{k-1}, \bar{P}_\gamma^{k-1})=\emptyset$, we have $(\tilde{T}_j, \tilde{T}_j+\tilde{P}_j(\tilde{T}_j))\cap(\bar{R}_\gamma, \bar{P}_\gamma)=\emptyset$. 

Thus, $\tilde{\mathbf{T}}$ is feasible, thereby implying $ans_3=yes$ in line~\ref{effsupervisor:T2_ans3} of Algorithm~\ref{algorithm:effsupervisor}. Since $\mathbf{T}_2 = \tilde{\mathbf{T}}$ exists, $\aggsafe{k+1,\infty}$ is well-defined.

Therefore, the supervisor is non-blocking.\qed
\end{proof}
\begin{figure*}[t!]
	\centering
	\begin{subfigure}{0.49\linewidth}
		\centering
		\includegraphics[width=\linewidth]{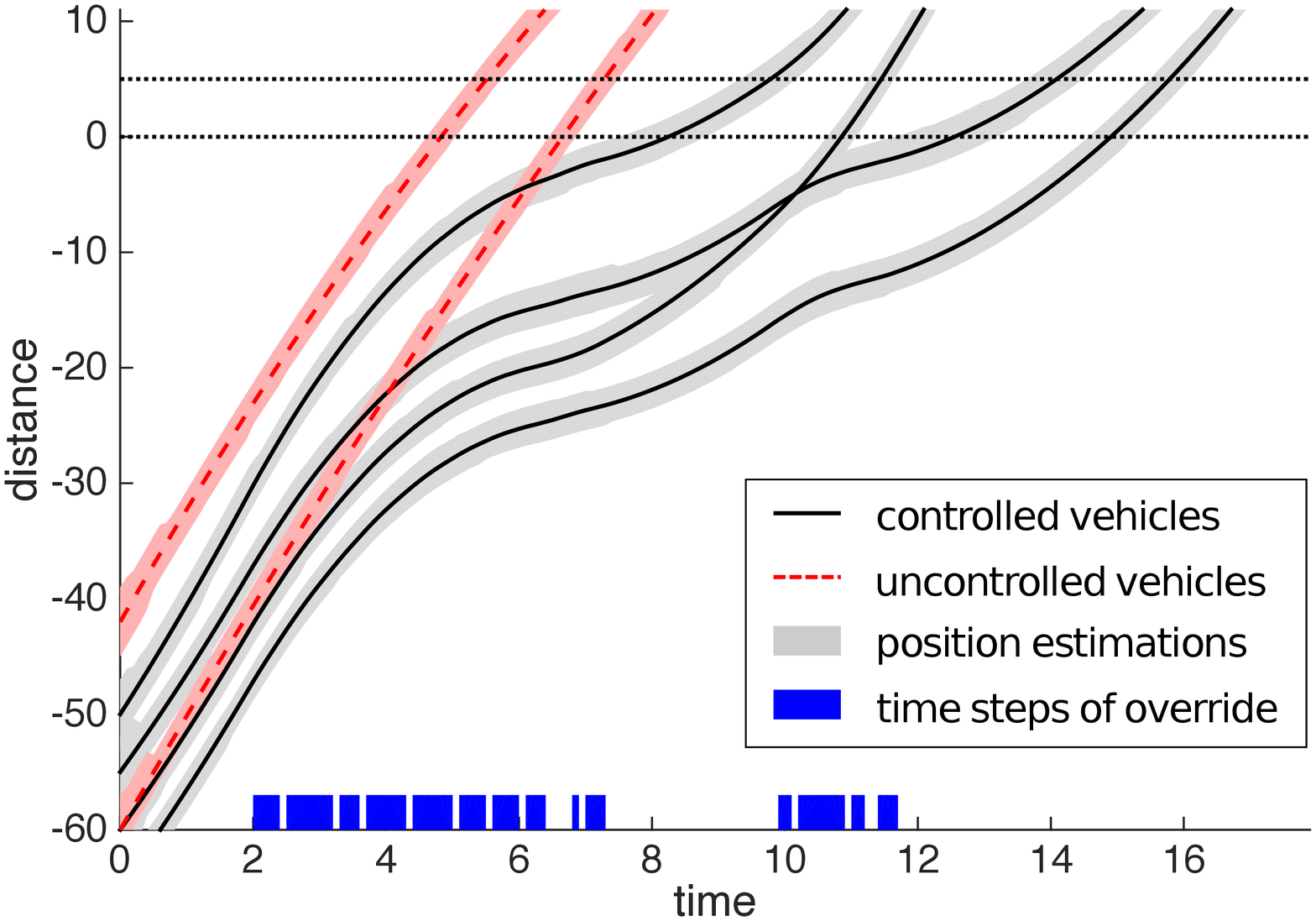}
		\caption{}
		\label{figure:exact_4veh}
	\end{subfigure}
	\begin{subfigure}{0.49\linewidth}
		\centering
		\includegraphics[width=\linewidth]{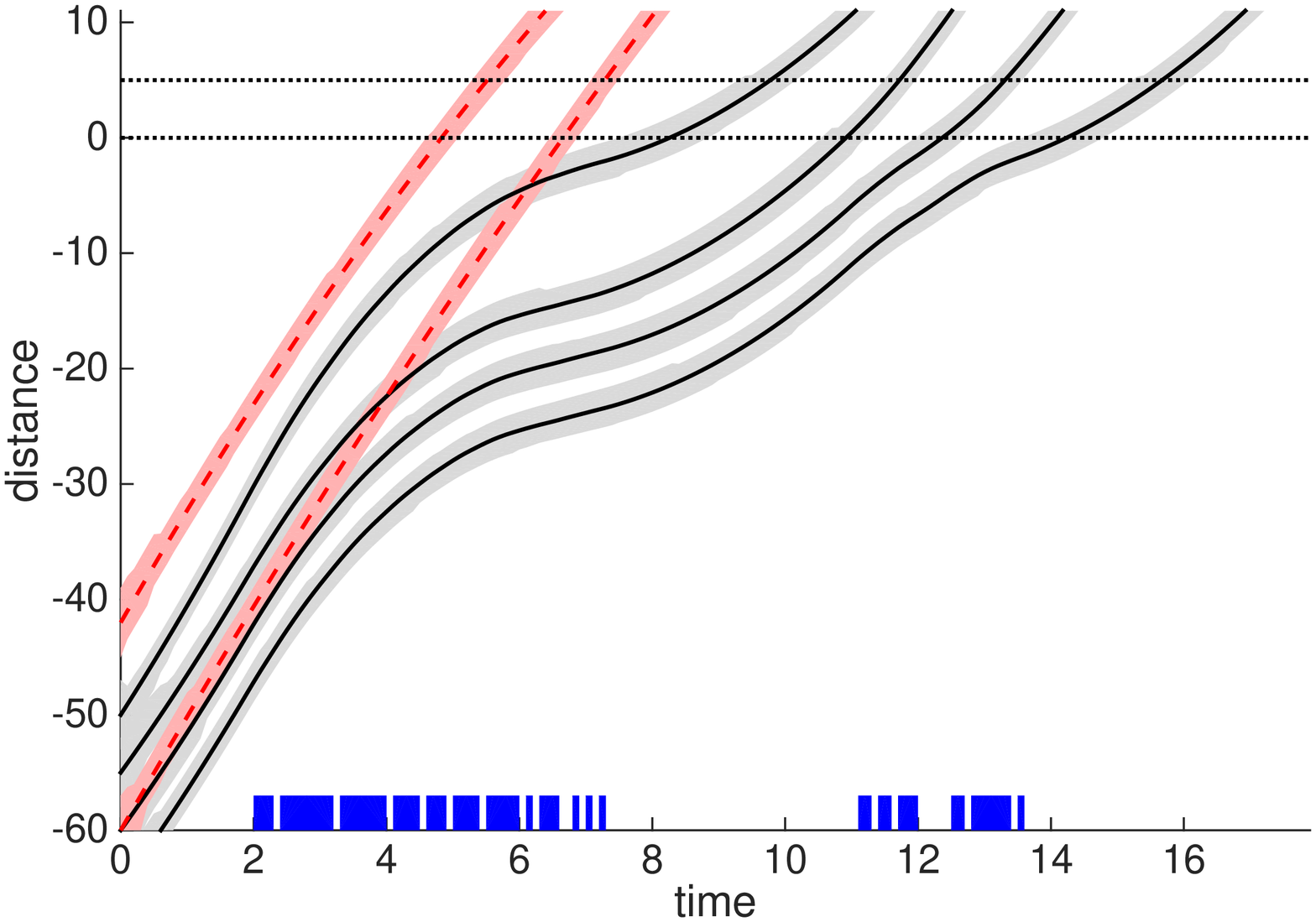}
		\caption{}
		\label{figure:approx_4veh}	
	\end{subfigure}
	\caption{Simulation of the exact supervisor $s$ in {(a)} and the efficient supervisor $s_e$ in {(b)}. The dotted lines at distance 0 and 5 represent the location of the intersection. The solid black lines and the dotted red lines are the exact trajectories of the controlled and uncontrolled vehicles, respectively, and the shading around the lines represent the position estimation. The blue boxes on the bottom are the time steps at which the supervisor overrides the controlled vehicles. The initial exact state is $\mathbf{y}(0)=$(-42, -50, -55, -60, -60, -65) and $\mathbf{v}(0)=$(10, 9, 8, 8, 10, 8).}
	\label{figure:simulation_4veh}
\end{figure*}

In summary, the efficient supervisors $\hat{s}$ and $s_e$ are more restrictive than the exact supervisor $s$, and $\hat{s}$ is more restrictive than $s_e$ by Theorem~\ref{theorem:efficientsupervisor}.

\section{Simulations}\label{section:simulation}
In this section, we present simulation results of the exact and efficient supervisors, $s$ and $s_e$, in two scenarios by running Algorithms~\ref{algorithm:supervisor} and \ref{algorithm:effsupervisor}. These are implemented using MATLAB on a personal computer with an 3.10\,GHz Intel Core i7-3770s processor with 8\,GB RAM.

With the dynamic states $x_j = (y_j, v_j)$ and $x_\gamma=(y_\gamma,v_\gamma)$ for $j\in\mathcal{C}$ and $\gamma\in\uncontrolled$, the vehicle dynamics considered in the simulation are as follows:
\begin{align*}
&\dot{y}_j = v_j+d_{y,j}, \\
&\dot{v}_j = \begin{cases} \max(0,u_j-b v_j^2 + d_{v,j}) & \text{if} ~(v_j=v_{j,min}), \\
\min(0,u_j-b v_j^2 + d_{v,j}) & \text{if} ~(v_j=v_{j,max}), \\
u_j - b v_j^2 + d_{v,j} & \text{otherwise},
\end{cases}
\end{align*}
and
\begin{align*}
&\dot{y}_\gamma = v_\gamma + d_{y,\gamma},\\
& \dot{v}_\gamma=\begin{cases} \max(0,w_\gamma-b v_\gamma^2 +d_{v,\gamma}) & 
\text{if} ~(v_\gamma=v_{\gamma,min}),\\
\min(0,w_\gamma-b v_\gamma^2 +d_{v,\gamma}) & 
\text{if} ~(v_\gamma=v_{\gamma,max}),\\
w_\gamma - b v_\gamma^2 + d_{v,\gamma} & \text{otherwise},
\end{cases}
\end{align*}
where $b = 0.001$ is a drag coefficient, and $d_{y,j}, d_{y,\gamma}, d_{v,i}, d_{v,\gamma}$ are disturbances representing unmodeled dynamics bounded by $-0.05$ and $0.05$. The speeds $v_j, v_\gamma$ are bounded by $v_{j,min}=v_{\gamma,min}=1.39$ and $v_{j,max}=v_{\gamma,max}=13.9$, the input $u_j$ by $u_{j,min}=-2.5$ and $u_{j,max}=2.5$, and the driver-input $w_\gamma$ by $w_{\gamma,min}=-0.5$ and $w_{\gamma,max}=0.5$. In the simulation, the disturbances and the driver-inputs of uncontrolled vehicles are randomly chosen within their bounds. 

At each time step $\tau=0.1$, the supervisors determine whether there are impending collisions at an intersection located at $(\alpha_i, \beta_i)=(0,5)$ for all $i\in\allset$. The states of all vehicles are measured subject to noises, $\delta_{y,i}\in [-3, 3]$ and $\delta_{v,i}\in [-0.05, 0.05]$. For the sake of simplicity, in the simulation, we let $u_{j,desire}^k=1$ for all vehicles at all times.



In the first scenario, four controlled and two uncontrolled vehicles are approaching an intersection as in Figure~\ref{figure:simulation_4veh}, which illustrates the position trajectories of the vehicles in time.

The simulation result of the exact supervisor $s$ is shown in Figure~\ref{figure:exact_4veh}, and that of the efficient supervisor $s_e$ is shown in Figure~\ref{figure:approx_4veh}. Though we said in Section~\ref{section:effsupervisor} that $s_e$ is more conservative than $s$ -- in the sense that given the same initial condition, $s$ evaluates more inputs to find one guaranteeing avoiding the Bad set than $s_e$ -- comparison of the override time steps (blue boxes) in Figures~\ref{figure:exact_4veh} and \ref{figure:approx_4veh} indicates that $s_e$ does not override the drivers more than $s$. It takes 0.067~s and 0.011~s per iteration to execute $s$ (Algorithm~\ref{algorithm:supervisor}) and $s_e$ (Algorithm~\ref{algorithm:effsupervisor}), respectively, in the worst case. This validates that the computation of the efficient supervisor is faster than that of the exact supervisor. In both simulations, the intersection at distance $(0,5)$ is occupied by one vehicle at a time.


\begin{figure}[htb!]
	\centering
	\includegraphics[width=\columnwidth]{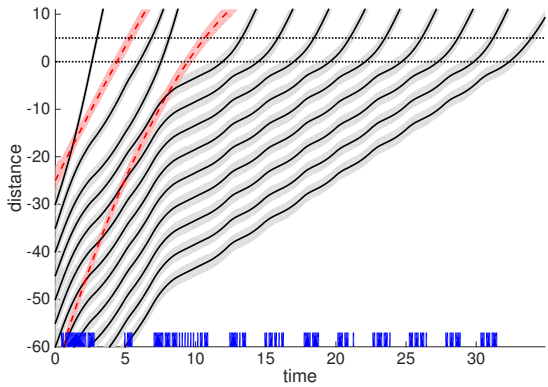}
	\caption{Simulation of the efficient supervisor $s_e$. The initial exact state is $\mathbf{y}(0)=$(-25, -30, -35, -40, -45, -50, -55, -60, -65, -65, -70, -75, -80, -85) and $\mathbf{v}(0)=$(6,9, 9, 8, 8, 8, 8, 8, 8, 9.5, 8, 8, 8, 8).}
	\label{figure:approx_12veh}
\end{figure}
	
The second scenario considers twelve controlled and two uncontrolled vehicles. Due to the large number of controlled vehicles involved, the exact supervisor cannot solve the verification problem within one time step. Thus, the only option to resolve conflict in this scenario is to implement the efficient supervisor (Algorithm~\ref{algorithm:effsupervisor}). As shown in Figure~\ref{figure:approx_12veh}, the efficient supervisor assists the vehicles to prevent any collision at the intersection. In the worst case, it takes 0.033~s per iteration.

\section{Experiments}\label{section:experiment}
In this section, we describe the experimental validation of the exact supervisor on an intersection testbed. First, we introduce the laboratory apparatus. Then, we describe the dynamic model of the RC cars used in the experiment and the techniques adopted to reduce disturbances. The results of the experiment are provided at the end of this section.

\subsection{Experimental setup}

\begin{figure}[htb!]
\centering
\begin{subfigure}{.49\linewidth}
  \centering
  \includegraphics[width=\linewidth, height=1.5 in]{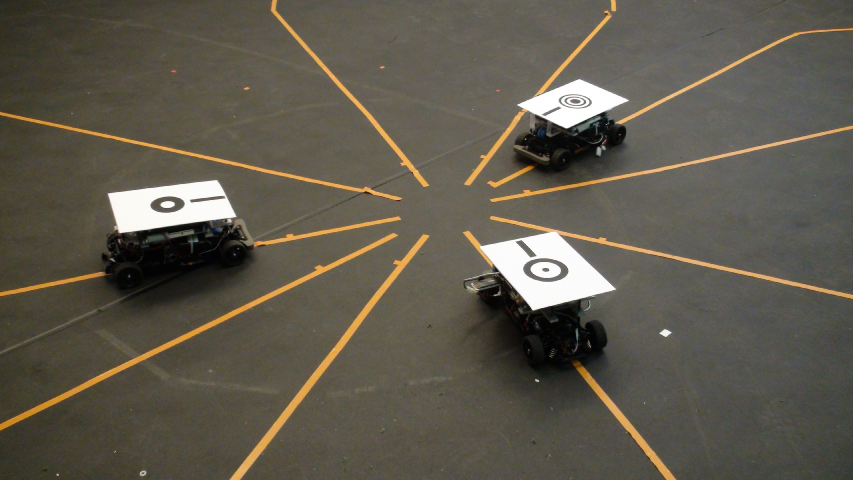}
  \caption{}
  \label{figure:carsandsymbols}
\end{subfigure}
\begin{subfigure}{.49\linewidth}
  \centering
  \includegraphics[width=\linewidth]{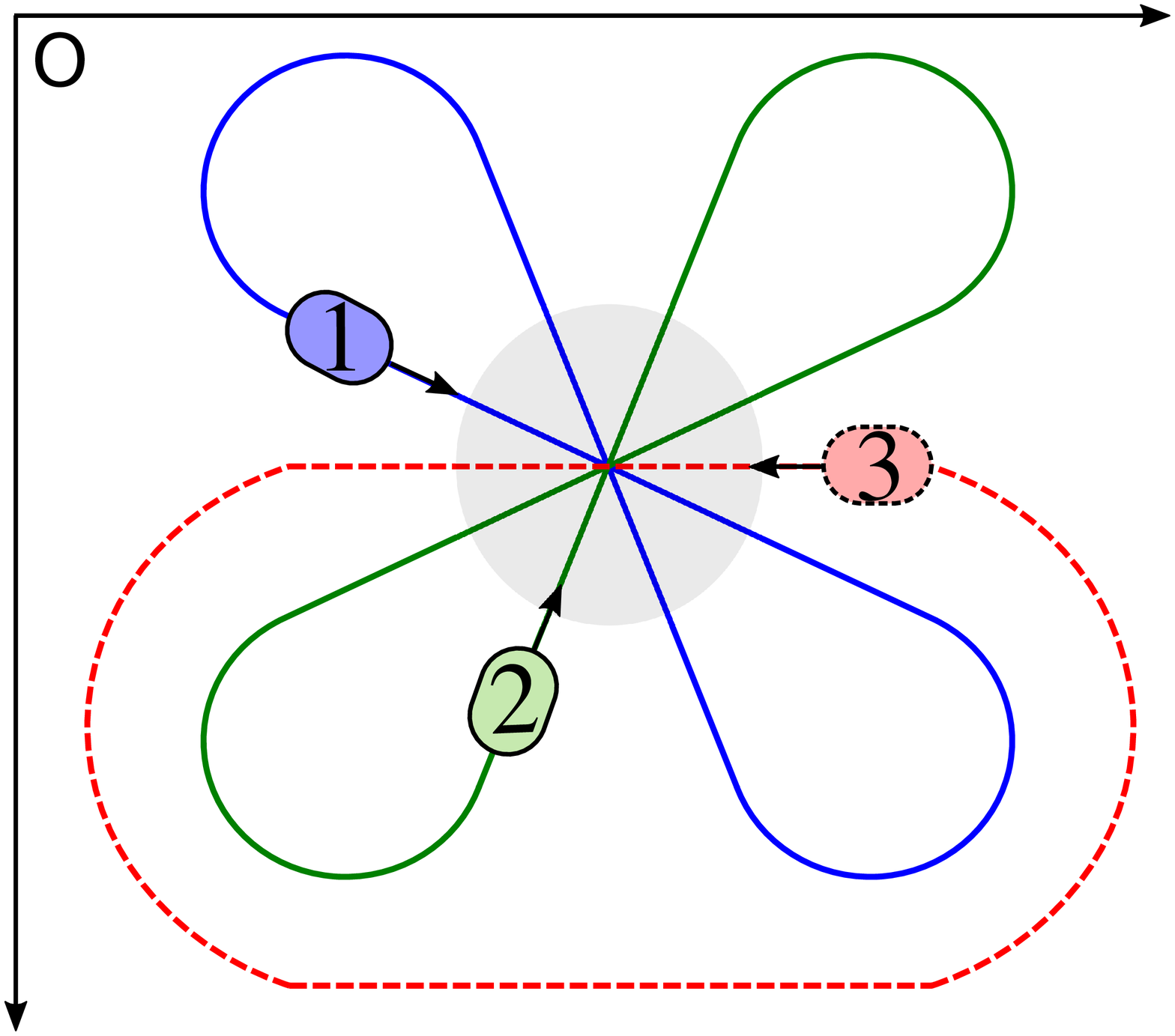}
  \caption{}
  \label{figure:paths}
\end{subfigure}
\caption{ \textbf{(a)} The RC cars used in the experiment. The camera system detects their symbols to identify each car and measure its position. \textbf{(b)} The paths of the cars on the testbed. Cars 1 and 2 are controlled by the supervisor while car 3 is not. A car is considered to be occupying the intersection when its position is within the shaded circle.}
\end{figure}

The cars used for this experiment are each built onto a Tamiya scaled RC car chassis equipped with a DC motor. A micro-controller (Acroname Moto 1.0) is used to control the steering servo and the motor through two separate PWM channels. An on-board computer (Mini-ITX running Fedora) runs the C programs, which provide the functionalities required to communicate with the centralized supervisor and to control the micro-controller. The system is powered by two batteries (Tenergy Li-Ion 14.8~V 4400~mAh) connected to a capacitor (Aluminium Electrolytic Capacitor 12000~uF 25~V) through a power relay (Omron G5SB).  A power amplifier (Acroname 3~A Back EMF H-bridge) connected to the batteries through a switch provides necessary power to the motor. 

During the experiment, three cars follow distinct paths intersecting at a single point on a 6~m $\times$ 6~m testbed as shown in Figure \ref{figure:paths}. Each path is stored as sequential points on the coordinate system illustrated in the figure. Cars 1 and 2 are controlled by the exact supervisor when necessary, while car 3 is not controllable. We program cars 1 and 2 to maintain constant motor input, which corresponds to the desired input, while car 3 is driven by a human operator. 


Each car has access to its own wheel speed through a quadrature encoder mounted on the rear axle. The position and direction of the cars are measured by an over-head vision system, which comprises six cameras on the ceiling and three computers processing images taken by the cameras. Each car is identified by a symbol attached to its roof as shown in Figure~\ref{figure:carsandsymbols}. The measurement of the speeds, positions, and directions of all cars are collected by an external computer via a wireless connection of the 802.11b standard using UDP/IP. Then, the computer distributes these measurements to the controlled cars, together with other information such as the desired inputs of the controlled vehicles and model parameters discussed in the next section. 

\subsection{Car dynamics model}\label{section:cardynamicsmodel}

\begin{figure}[htb!]
\centering
\includegraphics[width=\linewidth]{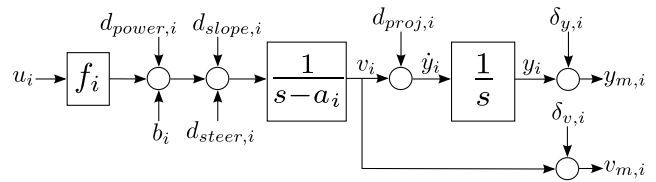}
\caption{Block diagram of the RC car model. The model includes the measurement disturbances $\delta_{y,i}$ and $\delta_{v,i}$ on the car position and speed, respectively.}
\label{figure:blockdiagram}
\end{figure}

A block diagram in Figure \ref{figure:blockdiagram} represents the car dynamics. In this section, we design a compensating input to reduce the effects of the disturbances $d_{proj,i}, d_{power,i}, d_{slope,i}$, and $d_{steer,i}$. 
 
Before compensating for the disturbances, the dynamics of the RC cars are as follows: for car $i\in\allset$ where $\mathcal{C}=\{1,2\}$ and $\uncontrolled=\{3\}$,
\begin{align}
\begin{split}\label{eq:rccarmodel}
	 &\dot{y}_i = v_i + d_{proj,i}, \\
	 &\dot{v}_i =	 a_i v_i + b_i + f_i u_i + d_{power,i} + d_{slope,i} + d_{steer,i} 
\end{split}
\end{align}
where ${y}_i$ is the longitudinal position of car $i$ along the path, $v_i$ is its wheel speed, and $u_i$ is the motor input with model parameters $a_i, b_i,$ and $f_i$. The disturbances $d_{proj,i}, d_{power,i}, d_{slope,i}$, and $d_{steer,i}$ are explained below and compensated to reduce their effects. 

\begin{figure}[htb!]
	\centering
	\begin{subfigure}{.49\linewidth}
		\centering
		\includegraphics[width=.6\linewidth]{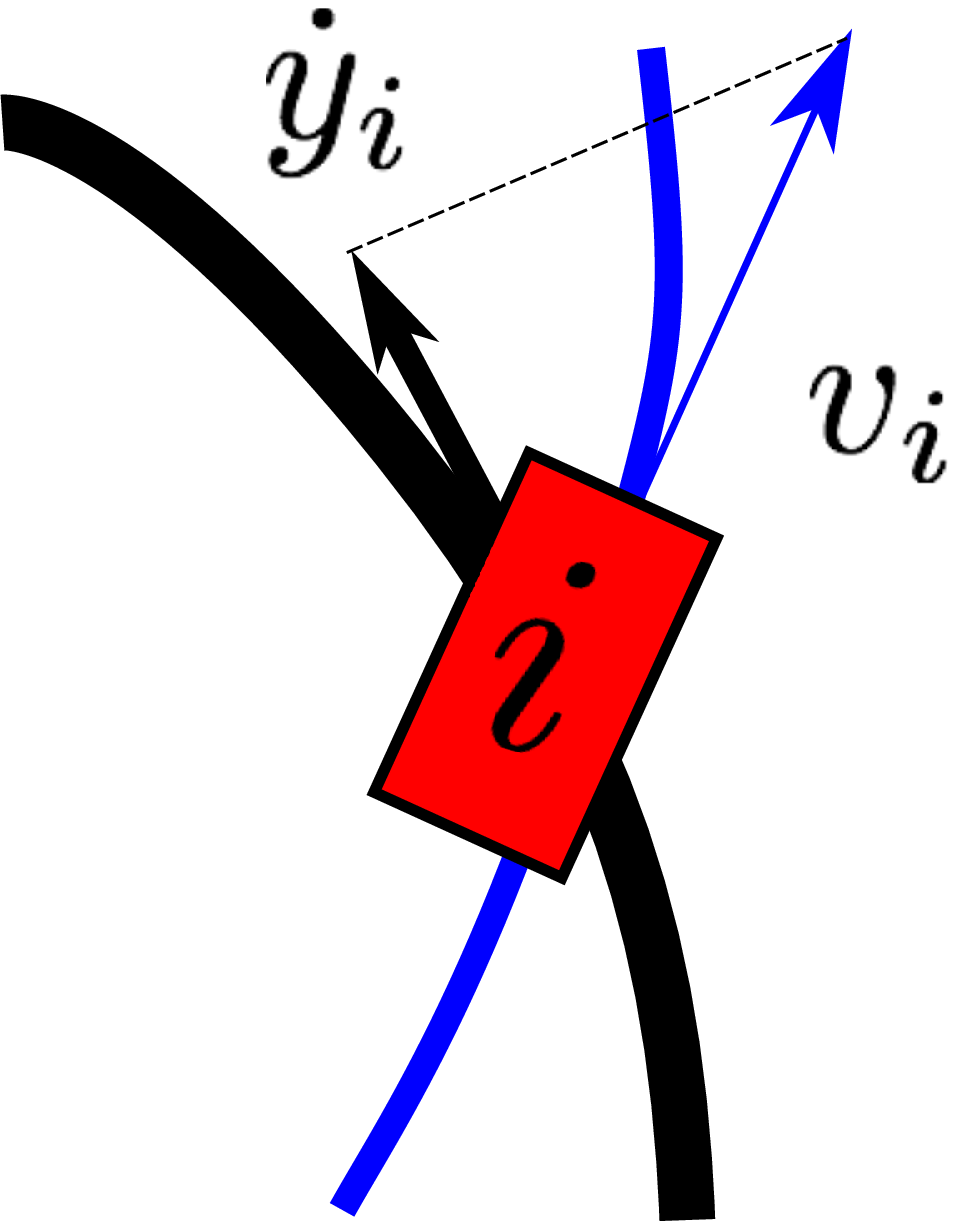}
		\caption{}
		\label{figure:dproja}
	\end{subfigure}%
	\begin{subfigure}{.49\linewidth}
		\centering
		\includegraphics[width=.7\linewidth]{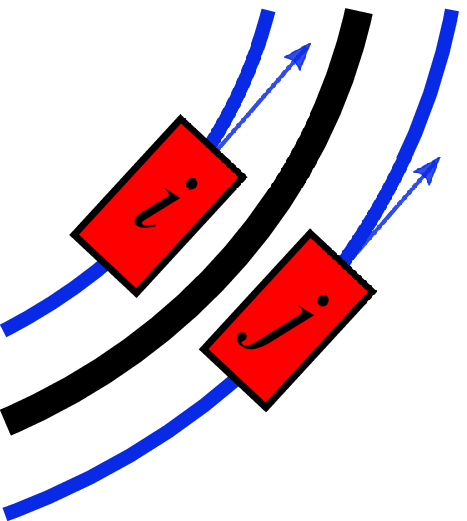}
		\caption{}
		\label{figure:dprojb}
	\end{subfigure}
	\caption{Cases when $d_{proj,i}\ne 0$. Imperfect path following contributes to the discrepancy between the longitudinal speed $\dot{y}_i$ and the wheel speed $v_i$.}
	\label{figure:dproj}
\end{figure}

Notice that $v_i$ is the car's wheel speed while $\dot{y_i}$ is the projected speed on the longitudinal path of car $i$. It is possible that $\dot{y_i}\ne v_i$ when the car does not follow its path exactly, and this discrepancy is represented by the term $d_{proj,i}$. Figure~\ref{figure:dproj} shows two cases in which $\dot{y}_i \ne v_i$. In Figure \ref{figure:dproja}, the direction of car $i$ differs from the tangent of the longitudinal path (black line) so that $v_i\ne \dot{y}_i$. In Figure~\ref{figure:dprojb}, two cars follow the same path with the same wheel speed $v$ but with slight deviation from the path, causing $\dot{y}_i >v > \dot{y}_j$.
\begin{figure*}[t!]
	\centering
	\begin{subfigure}{0.49\linewidth}
		\includegraphics[width = \linewidth]{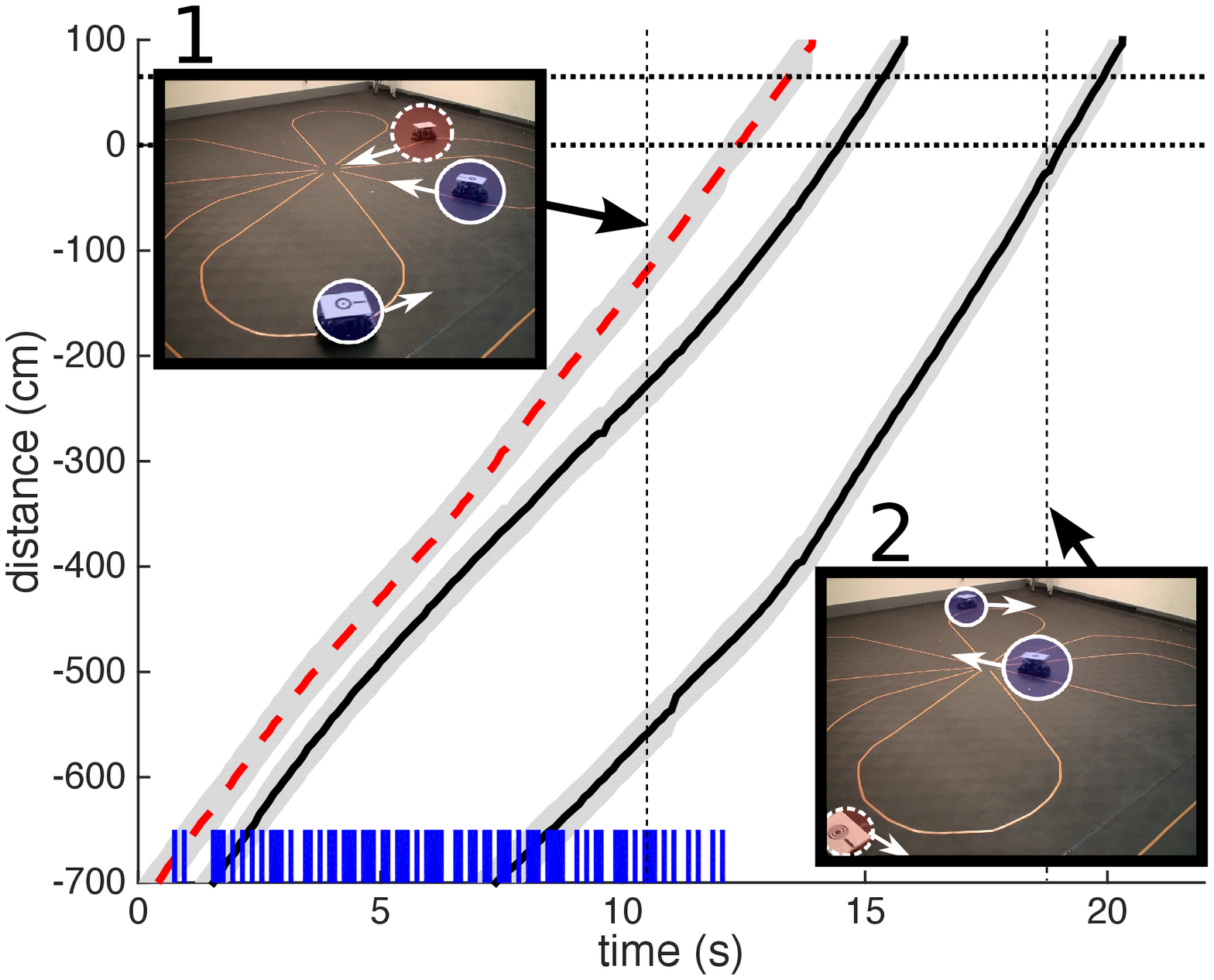}
		\caption{}
		\label{figure:timeposition1}
	\end{subfigure}
	\begin{subfigure}{0.49\linewidth}
		\includegraphics[width = \linewidth]{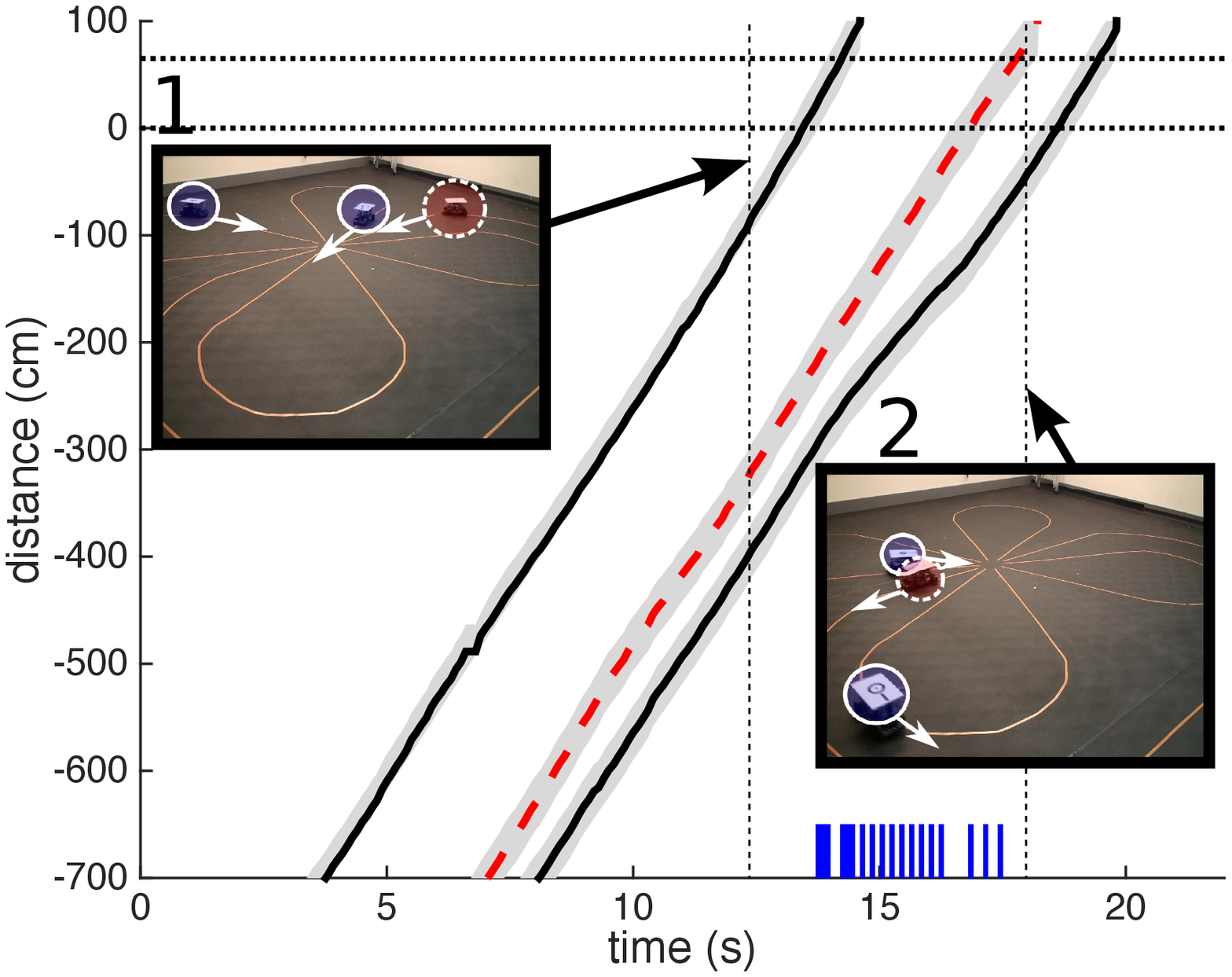}
		\caption{}
		\label{figure:timeposition2}
	\end{subfigure}
	\caption{The results of implementing the exact supervisor (see also Extension 1). The intersection is represented by the dotted lines located at $(0,65)$~cm. The solid black lines represent the position measurement of the controlled cars, cars 1 and 2, and the dotted red line represents the position measurement of the uncontrolled car, car 3. These figures demonstrate that the cars are not inside the intersection at the same time.}
	\label{figure:timegraphs}
\end{figure*}

This dynamic model \eqref{eq:rccarmodel} includes three different disturbances on the acceleration. The disturbance $d_{power,i}$ describes the first-order dynamic behavior empirically observed in the motor as a consequence of the power connection. By running the cars in a circle with constant motor input for several minutes, we can model $d_{power,i} = g_i e^{-t/h_i}u'_i$, where $g_i$ is the gain and $h_i$ the time constant. The car-specific parameters $g_i$ and $h_i$ are estimated by analyzing the collected data using the least square method (see \cite{rizzi_analysis_2014} for the data). 

The disturbance $d_{slope,i}$ is introduced to model the slope of the testbed, which is not completely flat and has non-negligible effects on the car speed. The disturbance $d_{steer,i}$ takes into account the fact that the steering and motor dynamics are coupled (\cite{verma_development_2008}). Since the testbed slope and the steering input are approximately the same at the same point of the path, we estimate these two disturbances as a path-dependent function $d_{path,i}(y) := d_{steer,i}(y) + d_{slope,i}(y)$. This function is different for each path and estimated by running the car multiple times and curve fitting of the obtained data.

Eliminating the effects of these disturbances is critical because otherwise it is difficult to initiate Algorithm~\ref{algorithm:supervisor} with a feasible initial condition, especially in a spatially constrained environment, such as a laboratory testbed. To this end, we introduce a compensating term $c_i(t, y_i)$ to the motor input so that $u_i = u'_i + c_i(t, y_i)$, where $u_i'$ is the input signal returned by the supervisor for car $i$. By employing the model and estimation of the disturbances explained above, we obtain the following compensating input:
$$ c_i(t, y_i) = - \frac{g_i e^{-t/h_i}u'_i + d_{path, i}(y_i)}{g_i e^{-t/h_i} + f_i}. $$
This, in turn, simplifies \eqref{eq:rccarmodel} as $\dot{y}_i=v_i+d_{proj,i}$ and $\dot{v}_i = a_i v_i + b_i + f_i u'_i$. Eventually with the compensation, we consider the following car dynamics in the experiment. For $i\in\allset$,
\begin{align*}
&\dot{y}_i = v_i+d_{y,i},\\
&\dot{v}_i=\begin{cases} 
\max(0,a_iv_i + b_i + f_i u_i' + d_{v,i})& \text{if}~(v_i = v_{i,min}),\\
\min (0,a_iv_i + b_i + f_i u_i' + d_{v,i})& \text{if}~(v_i = v_{i,max}),\\
a_i v_i +b_i +f_i u'_i + d_{v,i} & \text{otherwise},
\end{cases}
\end{align*}
where $d_{v,i}$ represents an error from the compensation and an unpredicted source of disturbance.

\subsection{Results}

\begin{figure*}[h!]
	\centering
	\includegraphics[width=\linewidth]{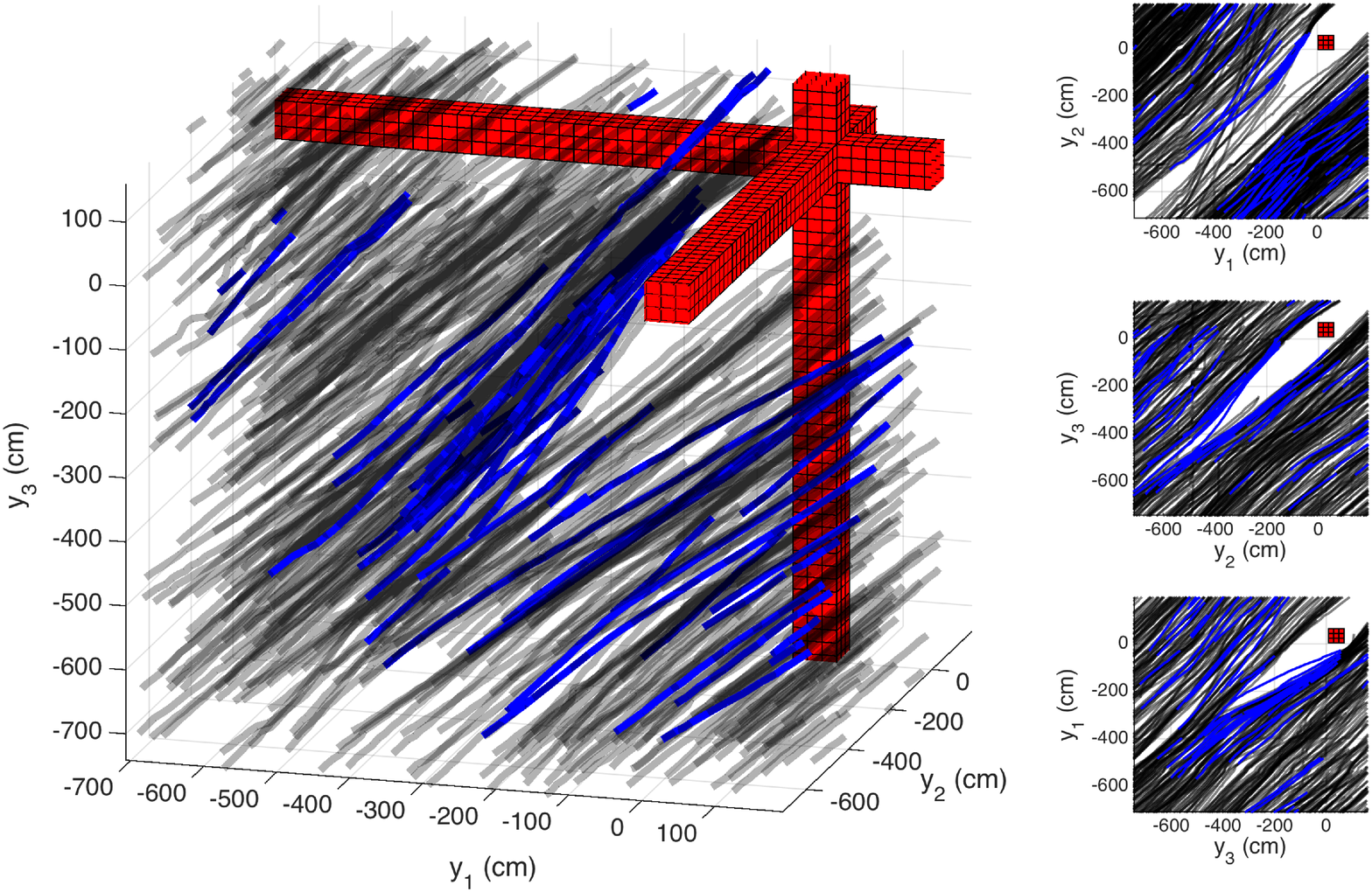}
	\caption{Trajectories (semitransparent black lines) in the output space. There are 509 trajectories, which are indicated in blue when the supervisor intervenes. The three figures on the right-hand side are the 2D projections. These confirm that all of the trajectories avoid the Bad set.}
	\label{figure:experiment}
\end{figure*}
The on-board computers on cars 1 and 2 run Algorithm \ref{algorithm:supervisor}. Since the same algorithm is run with the same measurement updated every $\tau=0.1 s$, the outputs are the same, thereby working as centralized control.

The parameters are as follows: $a =$(-0.53, -0.30, -0.43)~{s}$^{-1}$, $b =$(-84.68, -66.43, -49.64)~cm/s$^2$, $g=$(0.99, 0.44, 0.95), and $h =$(9.17, 6.84, 5.95)~s. The gain $f_i(t)$ for all $i$ is a time-varying parameter, estimated at every time step. The bounds of the speed are $v_{min}=$(10.5, 10.5, 13)~cm/s and $v_{max}=$(17, 16.5, 15)~cm/s, and those of the motor input (PWM) are $u_{min}=(105, 105, 130)$ and $u_{max}=$(170,165,150). The bounds of the measurement noises are empirically chosen by ignoring long tails as $\delta_{y,i,min}=$ -25~cm, $\delta_{y,i,max}=$ 25~cm, $\delta_{v,i,min}=$ -25~cm/s, and $\delta_{v,i,max}=$ 16~cm/s for all $i$. The bounds of the experimental disturbance are $d_{y,min}=$(-5, -3, -3)~cm/s, $d_{y,max}=$(3, 4, 3)~cm/s, $d_{v,min}=$(-4, -2, -3)~cm/s$^2$, and $d_{v,max}=$(2, 3, 2.5)~cm/s$^2$. 

Figure~\ref{figure:timegraphs} depicts two experimental results obtained by implementing the exact supervisor. The intersection is located at (0, 65)~cm for all cars. The intersection is an area containing the point at which the paths intersect, as in Figure~\ref{figure:paths}. 

In Figure \ref{figure:timeposition1}, in picture 1, the uncontrolled car (dotted red circle) approaches the intersection earlier than the other cars. The supervisor overrides the controlled cars (solid blue circles) to decelerate them until their desired inputs do not cause conflict. In picture 2, the conflict is resolved, and one controlled car crosses the intersection alone without overrides. In Figure \ref{figure:timeposition2}, in picture 1, one controlled car approaches the intersection first. The supervisor lets this car accelerate and the other controlled car decelerate so that the uncontrolled car safely crosses the intersection, as shown in picture 2. Notice that the upper bound of the last car's position enters the intersection right after the lower bound of the uncontrolled car's position has exited, indicating that the override was applied because it is deemed necessary to avoid the collision. In both cases, intersection collisions are averted.

Figure~\ref{figure:experiment} depicts 509 trajectories (semitransparent black lines) near the Bad set (red blocks). The trajectories are indicated in blue when the supervisor intervenes. From the 2D projections (figures on the right-hand side), we can confirm that none of the trajectories enters the Bad set. Since cars 1 and 2 are programmed to maintain constant speeds and car 3 does not change its speed quickly, the trajectories in the projections should be straight lines without overrides. We can see that the supervisor overrides cars 1 and 2 when the trajectories would enter the Bad set if the trajectories were linear. We observed in 15 trajectories that the cars were forced to stop before the Bad set because the supervisor could not find a safe input. This was because we had to truncate the tails of the distributions of disturbances, and thus the state measurement and the state estimation can sometimes be incompatible. This truncation is necessary in the confined laboratory because otherwise a feasible initial condition may not always exist. The existence of a feasible initial condition is a necessary condition to initiate procedure~\supervisor~in Algorithm~\ref{algorithm:supervisor}.


\section{Conclusions}\label{section:conclusions}
We have designed exact and efficient supervisors that override controlled vehicles when collisions are imminent. The sources of uncertainty, such as measurement errors, unmodeled dynamics, and the presence of uncontrolled vehicles are taken into account in the design of the supervisors. The exact supervisor determines the existence of safe inputs (verification problem) by solving the Inserted Idle-Time (IIT) scheduling problem, which is proven to yield equivalent answers to the verification problem. To address the computational complexity issue, we also design the efficient supervisor that solves the IIT scheduling problem with a quantified approximation bound. The simulation results show that the efficient supervisor prevents collisions without substantial conservatism, compared to the exact one. The experiment using RC cars on an intersection testbed validated that collisions at an intersection are successfully averted by the exact supervisor.

Although this paper deals only with decision problems to focus on safety, there is no barrier to incorporate objective functions to address other issues such as fuel consumption or traffic congestion. The intersection considered in this paper is modeled as a single conflict area so that vehicles are required to occupy the intersection one at a time. This assumption may make the system very conservative in that, for example, two vehicles turning right on different lanes are geometrically unable to collide while the supervisors do not let them inside the intersection at the same time. We are currently investigating the design of supervisors with a more general intersection model, which includes multi-conflict points. The result with simple first-order vehicle dynamics can be found in \cite{ahn_milp_2016}. Other remaining issues include preventing rear-end collisions (\cite{colombo_least_2014}) and considering unknown routes of vehicles.

\begin{funding}
This work was in part supported by NSF Award \#1239182. Alessandro Colombo was in part supported by grant AD14VARI02 - Sottomisura B.
\end{funding}

\bibliographystyle{SageH}
\bibliography{IEEEabrv}

\section*{Appendix A: Index to Multimedia Extensions}
\begin{table}[H]
\begin{tabular}{l l p{0.5\linewidth}}
	\hline
	Extension & Media type & Description \\
	\hline
	1 & Video & This video contains experiments of the exact supervisor presented in Section~\ref{section:experiment}.\\
	\hline
\end{tabular}
\end{table}

\section*{Appendix B}
We provide the proofs of Lemmas~\ref{lemma:garey_solve_iit}-\ref{lemma:Appro_no_Relaxed_no} in this section.

\noindent\textbf{Lemma~\ref{lemma:garey_solve_iit}.}
\textit{Procedure \garey~in Algorithm~\ref{algorithm:polynomial} solves the IIT scheduling problem with unit process times and finds a feasible schedule if exists.
}
\begin{proof}
To account for inserted idle-times, we define an initial set of forbidden regions as $F_{0,\gamma}=(\bar{r}_\gamma-1, \bar{p}_\gamma)$ for all $\gamma$. If $t_j\notin F_{0,\gamma}$ for some $j$, then we have $(t_j, t_j+1)\cap (\bar{r}_\gamma, \bar{p}_\gamma)=\emptyset$, thereby satisfying condition~\eqref{Unit-IIT:iit}.
	
Forbidden Region Declaration in lines~\ref{polynomial:forbidden_start}-\ref{polynomial:forbidden_end} of procedure \garey~solves Problem~\ref{problem:unit_process_scheduling}. The key idea is critical time $c$, which is the latest start time of job $\sigma_i$. If job $\sigma_i$ starts later than critical time $c$, at least one of the subsequent jobs cannot be scheduled before its deadline. That is, if a job cannot start before the critical time (line~\ref{polynomial:ans=no}), there is no schedule satisfying conditions \eqref{Unit-IIT:bounded}-\eqref{Unit-IIT:iit}. Otherwise, there will be a schedule satisfying all the conditions. This relation between critical time and the existence of a feasible schedule was proved in \cite{garey_scheduling_1981} with initially empty forbidden regions. Since initially non-empty forbidden regions do not affect the fact that critical time is the latest start time, this analysis of critical time can determine the existence of a schedule with non-empty initial forbidden regions.


With the forbidden regions declared, the EDD rule finds a schedule in lines~\ref{polynomial:initialize}-\ref{polynomial:end_while}. Time $s$ is assigned to each job as a schedule. In line~\ref{polynomial:forbidden}, $s$ avoids the forbidden regions so that condition~\eqref{Unit-IIT:iit} is satisfied.
If there is no ready job (line \ref{polynomial:if_A_empty_T=r}), a job with the minimum release time among unscheduled jobs is scheduled, and otherwise (line~\ref{polynomial:if_A_non_empty_T=s}), a job with the earliest deadline among ready jobs is scheduled. After a job is scheduled, line~\ref{polynomial:ABupdate} updates $s$ to $s+1$ so that condition~\eqref{Unit-IIT:disjoint} is satisfied. This schedule cannot satisfy condition~\eqref{Unit-IIT:bounded} if $ans$ has been $no$ in line~\ref{polynomial:ans=no}. Otherwise, this schedule satisfies conditions~\eqref{Unit-IIT:bounded}-\eqref{Unit-IIT:iit}.\qed
\end{proof}
\noindent\textbf{Lemma~\ref{lemma:T<barT}.}
\textit{If \procedureAS$([\lowx(0),\upperx(0)],S)=(\mathbf{T},yes)$, and 
\procedureRES$([\lowx(0),\upperx(0)],S)=(\bar{\mathbf{T}},\pi^*,yes)$, then
$T_{j}\leq \bar{T}_{j}$ for all $j\in\mathcal{C}$.}
\begin{proof}
We will show by induction on $j$ that $T_{\pi^*_j}\leq \bar{T}_{\pi^*_j} $ for all $\pi^*_j\in\mathcal{M}$. Notice that $T_{i}=\bar{T}_i=0$ for all $i\in\bar{\mathcal{M}}$. The schedule $\mathbf{T}$ is generated by procedure \procedureS~in Algorithm~\ref{algorithm:verification}. For the base case, $T_{\pi^*_1}=\max(R_{\pi^*_1},P_{max})$ in line~\ref{scheduling:i=1} of Algorithm~\ref{algorithm:verification}. Since $\bar{T}_{\pi^*_1}$ is a feasible solution of Problem~\ref{problem:relaxedS} by Lemma~\ref{lemma:garey_solve_iit}, it satisfies $R_{\pi^*_1}\leq\bar{T}_{\pi^*_1}$ from condition~\eqref{condition:efficient_boundedinput} and $(0,P_{i}(0))\cap (\bar{T}_{\pi^*_1}, \bar{T}_{\pi^*_1}+\theta_{max})=\emptyset$ for all $i\in\bar{\mathcal{M}}$ from condition~\eqref{condition:efficient_controlled}. Thus, $\max(R_{\pi^*_1},P_{max})=T_{\pi^*_1}\leq \bar{T}_{\pi^*_1}$. Now, suppose $T_{\pi^*_{k-1}} \leq \bar{T}_{\pi^*_{k-1}}$. Then, for $j=k$, we need to show that $T_{\pi^*_{k}}\leq \bar{T}_{\pi^*_{k}}$.
	
	In line~\ref{scheduling:con1and2} of Algorithm~\ref{algorithm:verification}, $T_{\pi^*_k}=\max(R_{\pi^*_k},T_{\pi^*_{k-1}}+P_{\pi^*_{k-1}}(T_{\pi^*_{k-1}}))$. If $T_{\pi^*_{k}}=R_{\pi^*_{k}}$, we have $T_{\pi^*_{k}}\leq \bar{T}_{\pi^*_{k}}$ because $\bar{T}_{\pi^*_k}$ satisfies condition~\eqref{condition:efficient_boundedinput}. If $T_{\pi^*_{k}}=T_{\pi^*_{k-1}}+P_{\pi^*_{k-1}}(T_{\pi^*_{k-1}})$, we have $T_{\pi^*_{k}}\leq \bar{T}_{\pi^*_{k-1}}+\theta_{max}$ because $T_{\pi^*_{k-1}}\leq \bar{T}_{\pi^*_{k-1}}$ and $P_{\pi^*_{k-1}}(T_{\pi^*_{k-1}})\leq \theta_{max}$ by \eqref{equation:themamax}. Since $\bar{T}_{\pi^*_{k}}$ satisfies  condition~\eqref{condition:efficient_controlled}, $\bar{T}_{\pi^*_{k-1}}+\theta_{max}\leq \bar{T}_{\pi^*_{k}}$. Therefore, $T_{\pi^*_{k}}\leq\bar{T}_{\pi^*_{k}}$. In lines~\ref{scheduling:Tj>Rbar_then>Pbar} and \ref{scheduling:Pj>Rbar_thenTj=Pbar} of Algorithm~\ref{algorithm:verification}, the schedule can increase so that $T_{\pi^*_{k}}=\bar{P}_\gamma$ for some $\gamma\in\uncontrolled$ if $T_{\pi^*_{k}}\geq \bar{R}_\gamma$ or $T_{\pi^*_{k}}+P_{\pi^*_{k}}(T_{\pi^*_{k}})> \bar{R}_\gamma$. Since $\bar{T}_{\pi^*_{k}}$ satisfies condition~\eqref{condition:efficient_uncontrolled}, if $\bar{T}_{\pi^*_{k}}\geq \bar{R}_\gamma$ or $\bar{T}_{\pi^*_{k}}+\theta_{max}> \bar{R}_\gamma$, then it must be $\bar{T}_{\pi^*_{k}}\geq \bar{P}_\gamma$. Therefore, in either case, $T_{\pi^*_{k}}\leq \bar{T}_{\pi^*_{k}}$. \qed
\end{proof}

\noindent\textbf{Lemma~\ref{lemma:Appro_no_Relaxed_no}.}
\textit{If \procedureAS$([\lowx(0),\upperx(0)],S)=(\emptyset,no)$, then \procedureRES$([\lowx(0),\upperx(0)],S)=(\emptyset,\pi^*,no)$.}

\begin{proof}
	In Algorithm~\ref{algorithm:approximate}, procedure \procedureAS~returns \textit{no} if $[\mathbf{y}^a(0),\mathbf{y}^b(0)]\cap B\ne\emptyset$ in line~\ref{approx:inside_B} or procedure \procedureS~with $\pi^*$ returns $no$ in line~\ref{approx:scheduling}. In the former case, procedure \procedureRES~also returns \textit{no} as in line~\ref{relaxed:inside_B} of Algorithm~\ref{algorithm:relaxedexact}. The latter case implies that $T_{\pi^*_j} > D_{\pi^*_j}$ for some $\pi^*_j\in\mathcal{C}$. By Lemma~\ref{lemma:T<barT}, we have $\bar{T}_ {\pi^*_j}\geq T_{\pi^*_j}$ if a feasible solution $\bar{\mathbf{T}}$ of Problem~\ref{problem:relaxedS} exists. Such a schedule cannot be feasible because $\bar{T}_{\pi^*_j} \geq T_{\pi^*_j} > D_{\pi^*_j}$, which does not satisfy condition~\eqref{condition:efficient_boundedinput}. Thus, procedure \procedureRES~returns $(\emptyset,\pi^*,no)$. \qed
\end{proof}

\end{document}